\documentclass[a4paper,UKenglish]{lipics-v2019}
\nolinenumbers

\usepackage{microtype}
\usepackage[dvipsnames]{xcolor}
\usepackage{tikz}

\renewcommand{\L}{\mathsf{L}}
\newcommand{\N}{\mathbb{N}}
\newcommand{\R}{\mathcal{R}}
\newcommand{\T}{\mathsf{T}}

\newcommand{\eps}{\varepsilon}

\newcommand{\wnd}{\mathrm{wnd}}
\newcommand{\rev}{\mathsf{R}}
\newcommand{\apx}{\mathsf{apx}}

\newcommand{\Conf}{\mathrm{Conf}}
\newcommand{\rConf}{\mathrm{rConf}}
\newcommand{\Rep}{\mathrm{Rep}}
\newcommand{\rep}{\mathsf{rep}}
\newcommand{\Suf}{\mathsf{Suf}}
\newcommand{\AllFlat}{\mathrm{AllFlat}}
\newcommand{\Flat}{\mathrm{Flat}}
\newcommand{\RegFlat}{\mathrm{RegFlat}}
\newcommand{\iml}{\mathsf{iml}}
\newcommand{\out}{\mathsf{out}}

\newcommand{\dom}{\mathrm{dom}}

\makeatletter
\DeclareRobustCommand{\cev}[1]{%
  \mathpalette\do@cev{#1}%
}
\newcommand{\do@cev}[2]{%
  \fix@cev{#1}{+}%
  \reflectbox{$\m@th#1\vec{\reflectbox{$\fix@cev{#1}{-}\m@th#1#2\fix@cev{#1}{+}$}}$}%
  \fix@cev{#1}{-}%
}
\newcommand{\fix@cev}[2]{%
  \ifx#1\displaystyle
    \mkern#23mu
  \else
    \ifx#1\textstyle
      \mkern#23mu
    \else
      \ifx#1\scriptstyle
        \mkern#22mu
      \else
        \mkern#22mu
      \fi
    \fi
  \fi
}

\makeatother

\title{Visibly Pushdown Languages over Sliding Windows}

\author{Moses Ganardi}{Universit\"at Siegen, Germany}{ganardi@eti.uni-siegen.de}{}{}

\authorrunning{Moses Ganardi}

\Copyright{Moses Ganardi}



\funding{The author is supported by the DFG project LO 748/13-1.}

\EventEditors{Rolf Niedermeier and Christophe Paul}
\EventNoEds{2}
\EventLongTitle{36th International Symposium on Theoretical Aspects of Computer Science (STACS 2019)}
\EventShortTitle{STACS 2019}
\EventAcronym{STACS}
\EventYear{2019}
\EventDate{March 13--16, 2019}
\EventLocation{Berlin, Germany}
\EventLogo{}
\SeriesVolume{126}
\ArticleNo{27}

\begin{document}

\maketitle

\begin{abstract}
	We investigate the class of visibly pushdown languages in the sliding window model.
	A sliding window algorithm for a language $L$ receives a stream of symbols and
	has to decide at each time step whether the suffix of length $n$ belongs to $L$ or not.
	The window size $n$ is either a fixed number (in the fixed-size model)
	or can be controlled by an adversary in a limited way (in the variable-size model).
	The main result of this paper states that for every visibly pushdown language
	the space complexity in the variable-size sliding window model
	is either constant, logarithmic or linear in the window size.
	This extends previous results for regular languages.
\end{abstract}

\section{Introduction}

\subparagraph{The sliding window model.}
A {\em sliding window algorithm (SWA)} is an algorithm which processes a stream of data elements $a_1 a_2 a_3 \cdots$
and computes at each time instant $t$ a certain value that depends on the suffix $a_{t-n+1} \cdots a_t$ of length $n$
where $n$ is a parameter called the {\em window size}.
This streaming model is motivated by the fact that in many applications data elements are outdated or become irrelevant after a certain time.
A general goal in the area of sliding window algorithms is to avoid storing the window content explicitly (which requires $\Omega(n)$ bits)
and to design space efficient algorithms, say using polylogarithmic many bits in the window size $n$.

A prototypical example of a problem considered in the sliding window model is the \textsc{Basic Counting} problem.
Here the input is a stream of bits and the task is to approximate the number of 1's in the last $n$ bits (the {\em active window}).
In \cite{DatarGIM02}, Datar, Gionis, Indyk and Motwani present an approximation algorithm using $O(\frac{1}{\epsilon} \log^2 n)$ bits of space
with an approximation ratio of $\epsilon$.
They also prove a matching lower bound of $\Omega(\frac{1}{\epsilon} \log^2 n)$ bits for any deterministic (and even randomized) algorithm
for \textsc{Basic Counting}.
Other works in the sliding window model include computing statistics \cite{ArasuM04,BabcockDMO03,BravermanO07},
optimal sampling \cite{BravermanOZ12}
and various pattern matching problems 
\cite{BreslauerG14,CliffordFPSS15,CliffordFPSS16,CliffordS16}.

There are two variants of the sliding window model, cf. \cite{ArasuM04}.
One can think of an adversary who can either insert a new element into the window
or remove the oldest element from the window.
In the {\em fixed-size} sliding window model the adversary determines the window size $n$ in the beginning
and the initial window is set to $a^n$ for some default known element $a$.
At every time step the adversary inserts a new symbol and then immediately removes the oldest element from the window.
In the {\em variable-size} sliding window model the window size is initially set to $n = 0$.
Then the adversary is allowed to perform an arbitrary sequence of insert- and remove-operations.
A remove-operation on an empty window leaves the window empty.
We also mention the timestamp-based model where every element carries a timestamp
(many elements may have the same timestamp)
and the active window at time $t$ contains only those elements whose timestamp
is at least $t-t_0$ for some parameter $t_0$ \cite{BravermanOZ12}.
Both the fixed-size and the timestamp-based model can be simulated in the variable-size model.

\subparagraph{Regular languages.}
In a recent series of works we studied the membership problem to a fixed regular language in the sliding window model.
It was shown in \cite{GanardiHL16} that in both the fixed-size and the variable-size sliding window model
the space complexity of any regular language is either constant, logarithmic or linear (a {\em space trichotomy}).
In a subsequent paper \cite{GHKLM18} a characterization of the space classes was given:
A regular language has a fixed/variable-size SWA with $O(\log n)$ bits if and only if
it is a finite Boolean combination of regular left ideals and regular length languages.
A regular language has a fixed-size SWA with $O(1)$ bits if and only if
it is a finite Boolean combination of suffix testable languages and regular length languages.
A regular language has a variable-size SWA with $O(1)$ bits if and only if
it is empty or universal.

\subparagraph{Context-free languages.}
A natural question is whether the results above can be extended to larger language classes,
say subclasses of the context-free languages.
More precisely, we pose the questions:
(i) Which language classes have a ``simple'' hierarchy of space complexity classes
(like the space trichotomy for the regular languages),
and (ii) are there natural descriptions of the space classes?
A positive answer to question (i) seems to be necessary to answer question (ii) positively.
In \cite{GJL18} we presented a family of context-free languages $(L_k)_{k \ge 1}$
which have space complexity $\Theta(n^{1/k})$ in the variable-size model
and $O(n^{1/k}) \setminus o(n^{1/k})$ in the fixed-size model,
showing that there exists an infinite hierarchy of space complexity classes
inside the class of context-free languages.
Intuitively, this result can be explained with the fact that a language and its complement have the same
sliding window space complexity;
however, the class of context-free languages is not closed under complementation (in contrast to the regular languages)
and the analysis of co-context-free languages in this setting seems to be very difficult.
Even in the class of deterministic context-free languages, which is closed under complementation,
there are example languages which have sliding window space complexity $\Theta((\log n)^2)$ \cite{GJL18}.

\subparagraph{Visibly pushdown languages.}
Motivated by these observations in this paper we will study the class of {\em visibly pushdown languages},
introduced by Alur and Madhusudan \cite{AlurM04}.
They are recognized by {\em visibly pushdown automata} where the alphabet is partitioned into {\em call letters}, {\em return letters}
and {\em internal letters}, which determine the behavior of the stack height.
Since visibly pushdown automata can be determinized,
the class of visibly pushdown languages turns out to be very robust
(it is closed under Boolean operations and other language operations)
and to be more tractable in many algorithmic questions than the class of context-free languages \cite{AlurM04}.
In this paper we prove a space trichotomy for the class of visibly pushdown languages
in the variable-size sliding window model, stating that the space complexity of every visibly pushdown language
is either $O(1)$, $\Theta(\log n)$ or $O(n) \setminus o(n)$.
The main technical result is a growth theorem (Theorem~\ref{thm:rational-dichotomy}) for rational transductions.
A natural characterization of the $O(\log n)$-class
as well as a study of the fixed-size model are left as open problems.

Let us mention some related work in the context of streaming algorithms for context-free languages.
Randomized streaming algorithms were studied for subclasses of context-free languages (DLIN and LL$(k)$) \cite{BabuLRV13}
and for Dyck languages \cite{MagniezMN14}.
A streaming property tester for visibly pushdown languages was presented by Fran{\c c}ois et al. \cite{FrancoisMRS16}.

\section{Preliminaries}

We define $\log n = \lfloor \log_2 n \rfloor$ for all $n \ge 1$,
which is the minimum number $k$ of bits
required to encode $n$ elements using bit strings of length {\em at most} $k$.
If $w = a_1 \cdots a_n$ is a word then any word of the form $a_i \cdots a_n$ ($a_1 \cdots a_i$)
is called {\em suffix (prefix)} of $w$.
A prefix (suffix) $v$ of $w$ is {\em proper} if $v \neq w$.
A {\em factor} of $w$ is any word of the form $a_i \cdots a_j$.
A {\em factorization} of $w$ is formally a sequence of possibly empty factors $(w_0, \dots, w_m)$
with $w = w_0 \cdots w_m$.
We call $w_0$ the {\em initial} factor and $w_1, \dots, w_m$ the {\em internal} factors.
The {\em reversal} of $w$ is $w^\rev = a_n a_{n-1} \cdots a_1$.
For a language $L \subseteq \Sigma^*$ we denote by $\Suf(L)$ the set of suffixes of words in $L$.
If $L = \Suf(L)$ then $L$ is {\em suffix-closed}.

\subparagraph{Automata.}
An {\em automaton} over a monoid $M$ is a tuple $A = (Q,M,I,\Delta,F)$
where $Q$ is a finite set of {\em states}, $I \subseteq Q$ is a set of {\em initial states},
$\Delta \subseteq Q \times M \times Q$ is the {\em transition relation} and
$F \subseteq Q$ is the set of {\em final states}.
A {\em run} on $m \in M$ from $q_0$ to $q_n$ is a sequence of transitions of the form
$\pi = (q_0,m_1,q_1)(q_1,m_2,q_2)\cdots (q_{n-1},m_n,q_n) \in \Delta^*$
such that $m = m_1 \cdots m_n$.
We usually depict $\pi$ as
$q_0 \xrightarrow{m_1} q_1 \xrightarrow{m_2} q_2 \, \cdots \, q_{n-1} \xrightarrow{m_n} q_n$,
or simply $q_0 \xrightarrow{m} q_n$.
It is {\em initial} if $q_0 \in I$ and {\em accepting} if $q_n \in F$.
The {\em language} defined by $A$ is the set $\L(A)$ of all elements $m \in M$
such that there exists an initial accepting run on $m$.
A subset $L \subseteq M$ is {\em rational} if $L = \L(A)$ for some automaton $A$.
We only need the case where $M$ is the free monoid $\Sigma^*$ over an alphabet $\Sigma$
or where $M$ is the product $\Sigma^* \times \Omega^*$ of two free monoids.
In these cases we change the format and write $(Q,\Sigma,I,\Delta,F)$
and $(Q,\Sigma,\Omega,I,\Delta,F)$, respectively.
Subsets of $\Sigma^*$ are called {\em languages} and subsets of $\Sigma^* \times \Omega^*$
are called {\em transductions}.
Rational languages are usually called {\em regular languages}.

In this paper we will also use {\em right automata}, which read the input from right to left.
Formally, a right automaton $A = (Q,M,F,\Delta,I)$ has the same format
as a (left) automaton where the sets of initial and final states are swapped.
Runs in right automata are defined from right to left,
i.e. a run on $m \in M$ from $q_n$ to $q_0$ is a sequence of transitions of the form
$(q_0,m_1,q_1)(q_1,m_2,q_2)\cdots (q_{n-1},m_n,q_n) \in \Delta^*$
such that $m = m_1 \cdots m_n$.
In the graphic notation we write the arrows from right to left.
It is initial (accepting) if $q_n \in I$ ($q_0 \in F$).

\subparagraph{Right congruences.}
For any equivalence relation $\sim$ on a set $X$ we write $[x]_\sim$ for the $\sim$-class containing $x \in X$
and $X/{\sim} = \{[x]_\sim \mid x \in X\}$ for the set of all $\sim$-classes.
The {\em index} of $\sim$ is the cardinality of $X/{\sim}$.
We denote by $\nu_\sim \colon X \to X/{\sim}$ the function with $\nu_\sim(x) = [x]_\sim$.
A subset $L \subseteq X$ is {\em saturated} by $\sim$ if $L$ is a union of $\sim$-classes.
An equivalence relation $\sim$ on the free monoid $\Sigma^*$ over some alphabet $\Sigma$
is a {\em right congruence}
if $x \sim y$ implies $xz \sim yz$ for all $x,y,z \in \Sigma^*$.
The {\em Myhill-Nerode right congruence} $\sim_L$ of a language $L \subseteq \Sigma^*$ is
the equivalence relation on $\Sigma^*$ defined by $x \sim_L y$ if and only if
$x^{-1} L = y^{-1} L$ where $x^{-1}L = \{ z \mid xz \in L \}$.
It is indeed the coarsest right congruence on $\Sigma^*$ which saturates $L$.
We usually write $\nu_L$ instead of $\nu_{\sim_L}$.
A language $L \subseteq \Sigma^*$ is regular iff $\sim_L$ has finite index.

\subparagraph{Rational transductions.}
Rational transductions are accepted by automata over $\Sigma^* \times \Omega^*$,
which are called finite state transducers.
In this paper, we will use a slightly extended but equivalent definition.
A {\em transducer} is a tuple $A = (Q,\Sigma,\Omega,I,\Delta,F,o)$
such that $(Q,\Sigma^* \times \Omega^*,I,\Delta,F)$ is an automaton over $\Sigma^* \times \Omega^*$
and a {\em terminal output function} $o \colon F \to \Omega^*$.
To omit parentheses we write runs $p \xrightarrow{(x,y)} q$ in the form $p \xrightarrow{x \mid y} q$
and depict $o(q) = y$ by a transition $q \xrightarrow{\mid y}$ without input word and target state.
If $\pi$ is a run $p \xrightarrow{x \mid y} q$ we define $\out(\pi) = y$
and $\out_F(\pi) = y \, o(q)$.
The transduction defined by $A$ is the set $\T(A)$ of all pairs $(x,\out_F(\pi))$
such that $\pi$ is an initial accepting run $p \xrightarrow{x \mid y} q$.
Since the terminal output function can be eliminated by $\eps$-transitions,
a transduction is rational if and only if it is of the form $\T(A)$ for some transducer $A$.
In this paper we will mainly use {\em rational functions}, which are
partial functions $t \colon \Sigma^* \to \Omega^*$ whose graph $\{(x,t(x)) \mid x \in \dom(t) \}$
is a rational transduction.

A transducer $A$ is {\em trim} if every state occurs on some accepting run.
If every word $x \in \Sigma^*$ has at most one initial accepting run
$p \xrightarrow{x \mid y} q$ for some $y \in \Omega^*$
then $A$ is {\em unambiguous}.
If $\Delta \subseteq Q \times \Sigma \times \Omega^* \times Q$
then $A$ is {\em real-time}.
It is known that every rational function is defined by a trim unambiguous real-time transducer \cite[Corollary~4.3]{Berstel79}.
If $A$ is unambiguous and trim then for every word $x \in \Sigma^*$ and every pair of states $(p,q) \in Q^2$ 
there exists at most one run from $p$ to $q$ with input word $x$.
Therefore, the state pair $(p,q)$ and the input word $x$ uniquely determine the run (if it exists)
and we can simply write $p \xrightarrow{x} q$.
	Similarly to \cite{WeberK95}, we define for a real-time transducer $A$ the parameter
	$\iml(A) = \max \left(\{ |y| \mid (q,a,y,p) \in \Delta \} \cup \{ |o(q)| \mid q \in Q \} \right)$.
	For every run $\pi$ on a word $x \in \Sigma^*$
	we have $|\out(\pi)| \le \iml(A) \cdot |x|$
	and $|\out_F(\pi)| \le \iml(A) \cdot (|x|+1)$.

The following closure properties for rational transductions are known \cite{Berstel79}:
The class of rational transductions is closed under inverse, reversal and composition
where the {\em inverse} of $T$ is $T^{-1} = \{ (y,x) \mid (x,y) \in T \}$,
the {\em reversal} of $T$ is $T^{\rev} = \{ (x^\rev,y^\rev) \mid (x,y) \in T \}$,
and the composition of two transductions $T_1,T_2$ is
$T_1 \circ T_2 = \{ (x,z) \mid \exists y: (x,y) \in T_1 \text{ and } (y,z) \in T_2 \}$.
If $T \subseteq \Sigma^* \times \Omega^*$ is rational and $L \subseteq \Sigma^*$ is regular
then the restriction $\{ (x,y) \in T \mid x \in L \}$ is also rational.
If $K \subseteq \Sigma^*$ is regular (context-free) and $T \subseteq \Sigma^* \times \Omega^*$ is rational
then $TK = \{ y \in \Omega^* \mid (x,y) \in T \text{ for some } x \in K \}$ 
is also regular (context-free).

A {\em right transducer} is a tuple $A = (Q,\Sigma,\Omega,F,\Delta,I,o)$
such that $(Q,\Sigma^* \times \Omega^*,F,\Delta,I)$ is a right automaton over $\Sigma^* \times \Omega^*$
and a {\em terminal output function} $o \colon F \to \Omega^*$.
We depict $o(q) = y$ by a transition $\xleftarrow{\mid y} q$.
If $\pi$ is a run $q \xleftarrow{x \mid y} p$ we define $\out(\pi) = y$
and $\out_F(\pi) = o(q) \, y$.
All other notions on transducers are defined for right transducers in a dual way.

\subparagraph{Growth functions.}

A function $\gamma \colon \N \to \N$ grows {\em polynomially}
if $\gamma(n) \in O(n^k)$ for some $k \in \N$;
we say that $\gamma$ grows {\em exponentially}
if there exists a number $c > 1$ such that $\gamma(n) \ge c^n$ for infinitely many $n \in \N$.
A function $\gamma(n)$ grows exponentially if and only if $\log \gamma(n) \notin o(n)$.

We will define a generalized notion of growth.
Let $t \colon \Sigma^* \to Y$ be a partial function
and let $X \subseteq \dom(t)$ be a language.
The {\em $t$-growth} of $X$
is the function $\gamma(n) = |t(X \cap \Sigma^{\le n})|$,
i.e. it counts the number of output elements on input words from $X$ of length at most $n$.
The {\em growth} of $X$ is simply the $\mathrm{id}_X$-growth of $X$,
i.e. $\gamma(n) = |X \cap \Sigma^{\le n}|$.
It is known that every context-free language has either polynomial or exponential growth \cite{Ginsburg1966}.
Furthermore, a context-free language $L$ has polynomial growth if and only if
it is {\em bounded}, i.e. $L \subseteq w_1^* \cdots w_k^*$ for some words $w_1, \dots, w_k$ \cite{Ginsburg1966}.
We need the fact that if $L$ is a bounded language
and $K$ is a set of factors of words in $L$ then $K$ is bounded \cite[Lemma~1.1(c)]{GS1964}.

\section{Visibly pushdown languages}

A {\em pushdown alphabet} is a triple $\tilde \Sigma = (\Sigma_c,\Sigma_r,\Sigma_{\mathit{int}})$
consisting of three pairwise disjoint alphabets:
a set of {\em call letters} $\Sigma_c$, a set of {\em return letters} $\Sigma_r$
and a set of {\em internal letters} $\Sigma_{\mathit{int}}$.
We identify $\tilde \Sigma$ with the union $\Sigma = \Sigma_c \cup \Sigma_r \cup \Sigma_{\mathit{int}}$.
The set of {\em well-matched} words $W$ over $\Sigma$ is defined as the smallest set which contains $\{\eps\} \cup \Sigma_{\mathit{int}}$, is closed under concatenation, and if $w$ is well-matched,
$a \in \Sigma_c$, $b \in \Sigma_r$ then also $awb$ is well-matched.
A word is called {\em descending (ascending)}
if it can be factorized into well-matched factors and return (call) letters.
The set of descending words is denoted by $D$.
A {\em visibly pushdown automaton (VPA)} has the form $A = (Q,\tilde \Sigma,\Gamma,\bot,q_0,\delta, F)$
where $Q$ is a finite state set, $\tilde \Sigma$ is a pushdown alphabet, $\Gamma$ is the finite stack alphabet
containing a special symbol $\bot$ (representing the empty stack), $q_0 \in Q$ is the initial state, $F \subseteq Q$ is the set of final states and $\delta = \delta_c \cup \delta_r \cup \delta_{\mathit{int}}$ is the transition function
where $\delta_c \colon Q \times \Sigma_c \to (\Gamma \setminus \{\bot\}) \times Q$,
$\delta_r \colon Q \times \Sigma_r \times \Gamma \to Q$
and $\delta_{\mathit{int}} \colon Q \times \Sigma_{\mathit{int}} \to Q$.
The set of {\em configurations} $\Conf$ is the set of all words $\alpha q$
where $q \in Q$ is a state and $\alpha \in  \bot (\Gamma \setminus \{\bot\})^*$
is the {\em stack content}.
We define $\delta \colon \Conf \times \Sigma \to \Conf$
for each $p \in Q$ and $a \in \Sigma$ as follows:
\begin{itemize}
	\item If $a \in \Sigma_c$ and $\delta(p,a) = (\gamma,q)$
	then $\delta(\alpha p,a) = \alpha \gamma q$.
	\item If $a \in \Sigma_{\mathit{int}}$ and $\delta(p,a) = q$
	then $\delta(\alpha p, a) = \alpha q$.
	\item If $a \in \Sigma_r$, $\delta(p,a,\gamma) = q$ and $\gamma \in \Gamma \setminus \{\bot\}$
	then $\delta(\alpha \gamma p, a) = \alpha q$.
	\item If $a \in \Sigma_r$ and $\delta(p,a,\bot) = q$
	then $\delta(\bot p) = \bot q$.
\end{itemize}
As usual we inductively extend $\delta$ to a function $\delta \colon \Conf \times \Sigma^* \to \Conf$
where $\delta(c,\eps) = c$ and $\delta(c,wa) = \delta(\delta(c,w),a)$ for all $w \in \Sigma^*$ and $a \in \Sigma$.
The {\em initial} configuration is $\bot q_0$ and a configuration $c$ is {\em final} if $c \in \Gamma^* F$.
A word $w \in \Sigma^*$ is {\em accepted} from a configuration $c$
if $\delta(c,w)$ is final.
The VPA $A$ {\em accepts} $w$ if $w$ is accepted from the initial configuration.
The set of all words accepted by $A$ is denoted by $\L(A)$;
the set of all words accepted from $c$ is denoted by $\L(c)$.
A language $L$ is a {\em visibly pushdown language (VPL)} if $L = \L(A)$ for some VPA $A$.
To exclude some pathological cases we assume that $\Sigma_c \neq \emptyset$ and $\Sigma_r \neq \emptyset$.
In fact, if $\Sigma_c = \emptyset$ or $\Sigma_r = \emptyset$ then any VPL over that pushdown alphabet would be regular.

One can also define nondeterministic visibly pushdown automata in the usual way,
which can always be converted into deterministic ones \cite{AlurM04}.
This leads to good closure properties of the class of all VPLs,
as closure under Boolean operations, concatenation and Kleene star.

The set $W$ of well-matched words forms a submonoid of $\Sigma^*$.
Notice that a VPA can only see the top of the stack when reading return symbols.
Therefore, the behavior of a VPA on a well-matched word
is determined only by the current state and independent of the current stack content.
More precisely, there exists a monoid homomorphism $\varphi \colon W \to Q^Q$ into the finite monoid
of all state transformations $Q \to Q$
such that $\delta(\alpha p, w) = \alpha \varphi(w)(p)$ for all $w \in W$ and $\alpha p \in \Conf$.

\section{Sliding window algorithms and main results}

In our context a {\em streaming algorithm} is a deterministic algorithm $A$
which reads an input word $a_1 \cdots a_m \in \Sigma^*$ symbol by symbol from left to right
and outputs after every prefix either $1$ or $0$.
We view $A$ as a deterministic (possibly infinite) automaton
whose states are encoded by bit strings and thus abstract away from the actual computation,
see \cite{GHKLM18} for a formal definition.
A {\em variable-size sliding window algorithm} for a language $L \subseteq \Sigma^*$ is a streaming algorithm $A$
which reads an input word $a_1 \cdots a_m$ over the extended alphabet $\overline{\Sigma} = \Sigma \cup \{ \downarrow \}$.
The symbol $\downarrow$ is the operation which removes the oldest symbol from the window.
At time $0 \le t \le m$ the algorithm has to decide whether the {\em active window} $\wnd(a_1 \cdots a_t)$ belongs to $L$
which is defined by
\begin{alignat*}{2}
\wnd(\varepsilon) & = \varepsilon  & \qquad     \wnd(u \! \downarrow) & = \varepsilon  \text{ if } \wnd(u) = \varepsilon \\
\wnd(ua) & = \wnd(u) a         &   \wnd(u \! \downarrow) & = v \text{ if } \wnd(u) = av
\end{alignat*}
for $u \in \Sigma^*$, $a \in \Sigma$.
For example, a variable-size sliding window algorithm $A$ for the language
$L_a = \{ w \in \{a,b\}^* \mid w \text{ contains } a \}$
maintains the window length $n$ and the position $i$ (from the right)
of the most recent $a$-symbol in the window (if it exists):
We initialize $n := 0$ and $i := \infty$.
On input $a$ we increment $n$ and set $i := 1$,
on input $b$ we increment both $n$ and $i$.
On input $\downarrow$ we decrement $n$, unless $n = 0$,
and then set $i := \infty$ if $i > n$.

The {\em space complexity} of $A$ is the function which maps $n$ to the maximum number of bits used
when reading an input $a_1 \cdots a_m$ where the window size never exceeds $n$, i.e. $|\wnd(a_1 \cdots a_t)| \le n$
for all $0 \le t \le n$.
Notice that this function is monotonic.
For every language $L$ there exists a space optimal variable-size sliding window algorithm \cite[Lemma~3.1]{GHKLM18a}
and we write $V_L(n)$ for its space complexity.
Clearly we have $V_L(n) \in O(n)$.
For example the example language $L_a$ above satisfies $V_{L_a}(n) \in O(\log n)$
because the algorithm above only maintains two numbers using $O(\log n)$ bits.
The main result of this paper states:

\begin{theorem}[Trichotomy for VPL]
	\label{thm:vpl-trichotomy}
	If $L$ is a visibly pushdown language then $V_L(n)$ is either $O(1)$, $\Theta(\log n)$ or
	$O(n) \setminus o(n)$.
\end{theorem}
In the rest of this section we will give an overview of the proof of Theorem~\ref{thm:vpl-trichotomy}.

\subparagraph{Suffix expansions.}

Let $\sim$ be an equivalence relation on $\Sigma^*$.
The {\em suffix expansion} of $\sim$ is the equivalence relation $\approx$ on $\Sigma^*$
defined by $a_1 \cdots a_n \approx b_1 \cdots b_m$ if and only if $n = m$ and
$a_i \cdots a_n \sim b_i \cdots b_n$ for all $1 \le i \le n$.
Notice that $\approx$ saturates each subset $\Sigma^{\le n}$.
Furthermore, if $\sim$ is a right congruence then so is $\approx$
since $|u| = |v|$ implies $|ua| = |va|$
and $a_i \cdots a_n \sim b_i \cdots b_n$ implies $a_i \cdots a_na \sim b_i \cdots b_na$.
We also define suffix expansions for partial functions $t \colon \Sigma^* \to Y$
with suffix-closed domain $\dom(t)$.
The {\em suffix expansion} of $t$
is the total function $\cev{t} \colon \dom(t) \to Y^*$ defined by
$\cev{t}(a_1 \cdots a_n) = t(a_1 \cdots a_n) \, t(a_2 \cdots a_n) \, \cdots \, t(a_{n-1} a_n) \, t(a_n)$
for all $a_1 \cdots a_n \in \Sigma^*$.
Here the range of $\cev{t}$ is the free monoid (alternatively, the set of all sequences) over $Y$.
If $\sim$ is an equivalence relation on $\Sigma^*$ then its suffix expansion $\approx$ is the {\em kernel}
of $\cev{\nu}_\sim$, i.e. $x \approx y$ if and only if $\cev{\nu}_\sim(x) = \cev{\nu}_\sim(y)$.
The space complexity in the variable-size model is captured by
the suffix expansion $\approx_L$ of the Myhill-Nerode right congruence $\sim_L$ or
alternatively by the suffix expansion $\cev{\nu}_L$ of $\nu_L$.

\begin{theorem}[{\cite[Theorem~4.1]{GHKLM18}}]
	\label{thm:psiL}
	For all $\emptyset \subsetneq L \subsetneq \Sigma^*$ we have
	$V_L(n) = \log |\Sigma^{\le n}/{\approx_L}| = \log |\cev{\nu}_L(\Sigma^{\le n})|$.
	In particular, $V_L(n) = \Omega(\log n)$ for every non-trivial language.
\end{theorem}
If $L$ is empty or universal, then $V_L(n) \in O(1)$
and otherwise $V_L(n) = \Omega(\log n)$.
Hence to prove Theorem~\ref{thm:vpl-trichotomy} it suffices to show that
either $V_L(n) \in O(\log n)$ or $V_L(n) \notin o(n)$ holds for every VPL $L$.
If $L$ is a regular language and $A$ is the minimal DFA of $L$ with state set $Q$,
one can identify $\nu_L(x)$ with the state $q \in Q$ reached on input $x$.
Hence, $\cev{\nu}_L(x)$ is represented by a word over $Q$.
Using the transition monoid of $A$ one can show that $\cev{\nu}_L \colon \Sigma^* \to Q^*$ is rational
(in fact {\em right-subsequential}, see Section~\ref{sec:rat-dichotomy})
and hence the image $\cev{\nu}_L(\Sigma^*) \subseteq Q^*$ is regular \cite[Lemma~4.2]{GHKLM18a}.
Since the growth of $\cev{\nu}_L(\Sigma^*)$ is either polynomial or exponential this implies
that $V_L(n) \in O(\log n)$ or $V_L(n) \notin o(n)$.

\subparagraph{Restriction to descending words.}
The approach above for regular languages can be extended to visibly pushdown languages $L$
if we restrict ourselves to the set $D$ of descending words.
If a VPA with state set $Q$ reads a descending word $x \in D$ from the initial configuration
it reaches some configuration $\bot q$ with empty stack.
Notice that there may be distinct configurations $\bot p \neq \bot q$
with $\L(\bot p) = \L(\bot q)$,
in which case we need to pick a single representative.
Since every suffix of $x$ is again descending
we can represent $\cev{\nu}_L(x)$ by a word $\sigma_0(x) \in Q^*$
and in fact we will prove that $S_0 = \sigma_0(D)$ is a context-free language (Lemma~\ref{lem:wmf-mon}).
By the growth theorem for context-free languages
the growth of $S_0$ is either polynomial or exponential.
If $S_0$ grows exponentially we obtain an exponential lower bound on
$|\cev{\nu}_L(\Sigma^{\le n})|$ (Lemma~\ref{lem:S-unbounded}).
Hence, the interesting case is that $S_0$ has polynomial growth,
i.e. $S_0$ is bounded.

\subparagraph{Representation by rational functions.}
In order to simulate a VPA by a finite automaton on arbitrary words
we will ``flatten'' the input word in the following way.
The input word $w$ is factorized $w = w_0 w_1 \cdots w_m$
into a descending prefix $w_0$,
and call letters and well-matched factors $w_1, \dots, w_m$.
The descending prefix $w_0$ is replaced by $\sigma_0(w_0)$ and each well-matched factor $w_i$
is replaced by a similar information $\sigma_1(w_i)$
which describes the behavior of the VPA on the factor $w_i$ and on each of its suffixes.
The set $\Flat$ of all flattenings is a context-free language.
Furthermore, there exists a rational function $\nu_f$ such that,
if a flattening $s$ represents a word $w \in \Sigma^*$
then $\nu_f(s)$ is a configuration representing the Myhill-Nerode class $\nu_L(w)$ (Proposition~\ref{prop:nu-f}).
Hence, we can reduce proving the main theorem
to the question whether the $\cev{\nu}_f$-growth of $\Flat$ is always either polynomial or exponential.

This question is resolved positively as follows.
We prove that for every rational function $t$ with suffix-closed domain $X = \dom(t)$
the $\cev{t}$-growth of $X$ is either polynomial or exponential (Theorem~\ref{thm:rational-dichotomy}).
In the case that $S_0$ has polynomial growth we can overapproximate $\Flat$ by a regular superset $\RegFlat$.
If the $\cev{\nu}_f$-growth of $\RegFlat$ is polynomial then the same holds trivially for the subset $\Flat$.
If the $\cev{\nu}_f$-growth of $\RegFlat$ is exponential then the proper choice of $\RegFlat$
ensures that $\Flat$ also has exponential $\cev{\nu}_f$-growth (Proposition~\ref{prop:transf-growth}).

\subparagraph{Dichotomy for rational functions.}
The main technical result of this paper states
that for every rational function $t \colon \Sigma^* \to \Omega^*$ 
with suffix-closed domain $X = \dom(t)$
the $\cev{t}$-growth of $X$ is either polynomial or exponential.
We emphasize that the range of $\cev{t}$ is not $\Omega^*$ but the free monoid over $\Omega^*$
(consisting of all finite sequences of words over $\Omega$).
There are in fact two reasons for exponential $\cev{t}$-growth: 
(i) The image $t(X)$ has exponential growth, and (ii) $X$ contains
a so called linear fooling set.
We need these lower bounds in the more general setting where $X \subseteq \dom(t)$ is a context-free subset,
namely $X = \Flat$.

\begin{proposition}
	\label{prop:exp-image}
	Let $t \colon \Sigma^* \to \Omega^*$ be rational with suffix-closed domain.
	If $X \subseteq \dom(t)$ is context-free
	and $t(X)$ has exponential growth then $X$ has exponential $t$-growth and
	exponential $\cev{t}$-growth.
\end{proposition}

\begin{example}
	\label{ex:rat-trans}
	Consider the transduction $f \colon \{a,b\}^* \to a^*$ defined by
	\[
		f = \{ (a^n,a^n) \mid n \in \N \} \cup \{ (a^n b w,a^n) \mid n \in \N, w \in \{a,b\}^* \},
	\]
	which projects a word over $\{a,b\}$ to its left-most (maximal) $a$-block and is rational.
	Its image $\cev{f}(\{a,b\}^*)$ can be identified with the set of all sequences
	of natural numbers which are concatenations of monotonically decreasing sequences of the form $(k,k-1, \dots, 0)$.
	There are exactly $2^n$ of such sequences of length $n$
	and hence $\{a,b\}^*$ has exponential $\cev{f}$-growth.
\end{example}

A {\em linear fooling scheme} for a partial function $t \colon \Sigma^* \to Y$
is a tuple $(u_2,v_2,u,v,Z)$ where $u_2,v_2,u,v \in \Sigma^*$ and $Z \subseteq \Sigma^*$ such that
$u_2$ is a suffix of $u$ and $v_2$ is a suffix of $v$, $|u_2| = |v_2|$,
$\{u_2,v_2\}\{u,v\}^*Z \subseteq \dom(t)$ and
for all $n \in \N$ there exists a word $z_n \in Z$ of length $|z_n| \le O(n)$
such that $t(u_2 w z_n) \neq t(v_2 w z_n)$ for all $w \in \{u,v\}^{\le n}$.
The set $\{u_2,v_2\}\{u,v\}^*Z$ is called a {\em linear fooling set} for $t$.
Notice that the definition implies that $u_2 \neq v_2$ and hence $u$ is not a suffix of $v$,
and vice versa, i.e. $\{u,v\}$ is a {\em suffix code}.
Therefore $\{u,v\}^n$ contains $2^n$ words of length $O(n)$
and thus $\{u_2,v_2\}\{u,v\}^*$ has exponential growth.

\begin{proposition}
	\label{prop:fooling}
	Let $t \colon \Sigma^* \to \Omega^*$ be a partial function with suffix-closed domain.
	If $X \subseteq \dom(t)$ contains a linear fooling set for $t$
	then the $\cev{t}$-growth of $X$ is exponential.
\end{proposition}

\begin{proof}
	Let $(u_2,v_2,u,v,Z)$ be a linear fooling scheme
	with $\{u_2,v_2\}\{u,v\}^*Z \subseteq X$.
	Let $n \in \N$ and let $z_n \in Z$ with the properties from the definition.
	Consider two distinct words $w, w' \in \{u,v\}^n$.
	Without loss of generality the words have the form $w = w_1 u w_2$ and $w' = w_3 v w_2$
	for some $w_1,w_2,w_3 \in \{u,v\}^*$.
	Hence $w$ has the suffix $u_2 w_2$ and $w'$ has the suffix $v_2 w_2$, which are suffixes of the same length.
	By assumption we have $t(u_2 w_2 z_n) \neq t(v_2 w_2 z_n)$ and hence also $\cev{t}(w z_n) \neq \cev{t}(w' z_n)$.
	This implies that $|\cev{t}(u_2 \{u,v\}^n z_n)| \ge 2^n$ for all $n \in \N$.
	Since all words in $u_2 \{u,v\}^n z_n \subseteq X$ have length $O(n)$
	there exists a number $c > 1$ such that
	$|\cev{t}(X \cap \Sigma^{\le cn})| \ge 2^n$ for sufficiently large $n$.
\end{proof}

The following dichotomy theorem will be proved in Section~\ref{sec:rat-dichotomy}.
\begin{theorem}
	\label{thm:rational-dichotomy}
	Let $t \colon \Sigma^* \to \Omega^*$ be rational and with suffix-closed domain $X = \dom(t)$.
	If $X$ contains no linear fooling set for $t$ and $t(X)$ is bounded
	then the $\cev{t}$-growth of $X$ is polynomial.
	Otherwise the $\cev{t}$-growth of $X$ is exponential.
\end{theorem}

\section{Reduction to transducer problem}

Fix a VPA $A = (Q,\tilde \Sigma,\Gamma,\bot,q_0,\delta, F)$
and let $\emptyset \subsetneq L = \L(A) \subsetneq \Sigma^*$ for the rest of this section.

\subparagraph{Monotonic factorization.}
A factorization of $w = w_0 w_1 \cdots w_m \in \Sigma^*$ into factors $w_i \in \Sigma^*$
is {\em monotonic} if
$w_0$ is descending (possibly empty) and for each $1 \le i \le m$ the factor $w_i$ is either
a call letter $w_i \in \Sigma_c$ or a non-empty well-matched factor.
If $w_0 w_1 \cdots w_m$ is a monotonic factorization
then $w_i' w_{i+1} \cdots w_j$ is a monotonic factorization
for any $0 \le i \le j \le m$ and suffix $w_i'$ of $w_i$.
To see that every word $w \in \Sigma^*$ has at least one monotonic factorization
consider the set of non-empty maximal well-matched factors in $w$ (maximal with respect to inclusion).
Observe that two distinct maximal well-matched factors in a word cannot overlap
because the union of two overlapping well-matched factors is again well-matched.
Since every internal letter is well-matched the remaining positions contain only return and call letters.
Furthermore, every remaining call letter must be to the right of every remaining return letter,
which yields a monotonic factorization of $w$.
Figure~\ref{fig:sc} shows a monotonic factorization $w = w_0 w_1 \cdots w_m$
where the descending prefix $w_0$ is colored red and call letters $w_i$ are colored green.
The {\em stack height function} for the word $w$ increases (decreases) by one on call (return) letters and stays constant on internal letters.

\begin{figure}

\centering

\scalebox{0.6}{
\begin{tikzpicture}

\tikzstyle{n} = [circle, fill, inner sep = 1.5pt]
\tikzstyle{r} = [red]

\draw[step=1.0,black!20,thin] (0,0) grid (20,4);
\draw (0,2) node[n] {} -- (1,1) node[n] {} -- (2,1) node[n] {}
-- (3,2) node[n] {} -- (4,1) node[n] {} -- (5,0) node[n] {}
-- (6,0) node[n] {} -- (7,1) node[n] {} -- (8,0) node[n] {}
-- (9,1) node[n] {} -- (10,2) node[n] {} -- (11,1) node[n] {}
-- (12,1) node[n] {} -- (13,2) node[n] {} -- (14,3) node[n] {}
-- (15,4) node[n] {} -- (16,3) node[n] {} -- (17,4) node[n] {}
-- (18,3) node[n] {} -- (19,2) node[n] {} -- (20,3) node[n] {};
\node at (0.5,-.6) {\Huge \textcolor{red}{$b$}};
\node at (1.5,-.7) {\Huge \textcolor{red}{$c$}};
\node at (2.5,-.7) {\Huge \textcolor{red}{$a$}};
\node at (3.5,-.6) {\Huge \textcolor{red}{$b$}};
\node at (4.5,-.6) {\Huge \textcolor{red}{$b$}};
\node at (5.5,-.7) {\Huge $c$};
\node at (6.5,-.7) {\Huge $a$};
\node at (7.5,-.6) {\Huge $b$};
\node at (8.5,-.7) {\Huge \textcolor{Cerulean}{$a$}};
\node at (9.5,-.7) {\Huge $a$};
\node at (10.5,-.6) {\Huge $b$};
\node at (11.5,-.7) {\Huge $c$};
\node at (12.5,-.7) {\Huge \textcolor{Cerulean}{$a$}};
\node at (13.5,-.7) {\Huge $a$};
\node at (14.5,-.7) {\Huge $a$};
\node at (15.5,-.6) {\Huge $b$};
\node at (16.5,-.7) {\Huge $a$};
\node at (17.5,-.6) {\Huge $b$};
\node at (18.5,-.6) {\Huge $b$};
\node at (19.5,-.7) {\Huge \textcolor{Cerulean}{$a$}};

\end{tikzpicture}
}
\caption{The stack height function for a word ($\Sigma_c = \{a\}$, $\Sigma_r = \{b\}$, $\Sigma_{\mathit{int}} = \{c\}$)
and a monotonic factorization $bcabb ~ cab ~ a ~ abc ~ a ~ aababb ~ a$.
}
\label{fig:sc}
\end{figure}
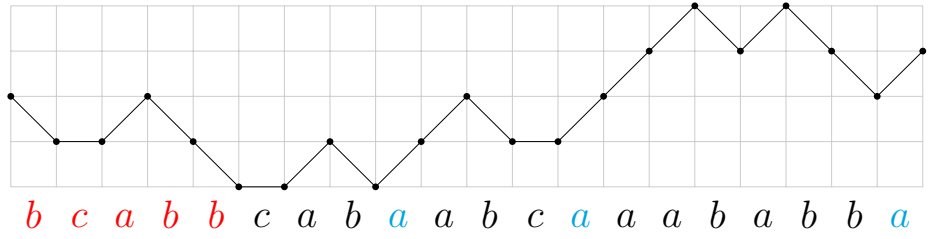

\subparagraph{Representation of Myhill-Nerode classes.}
To apply Theorem~\ref{thm:psiL} we need a suitable description of the $\sim_L$-classes.
We follow the approach in \cite{BaranyLS06} of choosing length-lexicographic minimal
representative configurations.
Since their definition slightly differs from ours
(according to their definition, a VPA may not read a return letter if the stack contains $\bot$ only)
we briefly recall their argument (in the appendix).
Let $\rConf = \{ \delta(\bot q_0,w) \mid w \in \Sigma^* \}$
be the set of all {\em reachable} configurations in $A$,
which is known to be regular \cite{BouajjaniEM97,Buechi64}.
Two configurations $c_1,c_2 \in \rConf$ are {\em equivalent}, denoted by $c_1 \sim c_2$, if
$\L(c_1) = \L(c_2)$.
By fixing arbitrary linear orders on $\Gamma$ and $Q$ we can consider the length-lexicographical order
on $\rConf$ and define the function $\rep \colon \rConf \to \rConf$
which chooses the minimal representative from each $\sim$-class,
i.e. for all $c \in \rConf$ we have $\rep(c) \sim c$
and for any $c' \in \rConf$ with $c \sim c'$ we have $\rep(c) \le_{\mathit{llex}} c'$.
The set of representative configurations is denoted by $\Rep = \rep(\rConf)$.
\begin{lemma}[\cite{BaranyLS06}]
	\label{lem:rep-rat}
	The function $\rep$ is rational.
\end{lemma}
Finally we define $\nu_A \colon \Sigma^* \to \Rep$ by $\nu_A(w) = \rep(\delta(\bot q_0,w))$
for all $w \in \Sigma^*$.
It represents $\sim_L$ in the sense that
$\L(\nu_A(w)) = w^{-1} \L(A)$ for all $w \in \Sigma^*$
and hence $\nu_A(u) =\nu_A(v)$ if and only if $u \sim_L v$.
Therefore we have $V_L(n) = \log |\cev{\nu}_A(\Sigma^{\le n})|$
by Theorem~\ref{thm:psiL}.

\subparagraph{Flattenings.}
Since we cannot compute $\nu_A$ using a finite state transducer
we choose a different representation of the input.
Define the alphabet $\Sigma_f = \Sigma_c \cup Q \cup Q^Q$.
A {\em flattening} is a word $s_0 s_1 \cdots s_m \in \Sigma_f^*$
where $s_0 \in Q^*$ and $s_i \in \Sigma_c \cup Q^Q Q^*$ for all $1 \le i \le m$.
Notice that the factorization $s = s_0 s_1 \cdots s_m$ is unique.
The set of all flattenings is $\AllFlat = Q^* (\Sigma_c \cup Q^Q Q^*)^*$.
We define a function $t_f \colon \AllFlat \to \rConf$ as follows.
Let $s = s_0 s_1 \cdots s_m \in \Sigma_f^*$ be a flattening and
we define $t_f(s)$ by induction on $m$:
\begin{itemize}
\item If $s_0 = \eps$ then $t_f(s_0) = \bot q_0$.
If $s_0 = q_1 \cdots q_n \in Q^+$ then $t_f(s_0) = \bot q_1$.
\item If $s_m \in \Sigma_c$ then $t_f(s_0 \cdots s_m) = \delta(t_f(s_0 \cdots s_{m-1}), s_m)$.
\item If $s_m = \tau q_2 \cdots q_m \in Q^Q Q^*$ and $t_f(s_0 \cdots s_{m-1}) = \alpha q$
then $t_f(s) = \alpha \tau(q)$.
\end{itemize}
Define the function $\nu_f \colon \AllFlat \to \Rep$ by $\nu_f = \rep \circ t_f$.
\begin{lemma}
\label{lem:t-f-rat}
The functions $t_f$ and $\nu_f$ are rational.
\end{lemma}
\begin{proof}
We first define a transducer $A_1$ which handles flattenings where the initial factor is empty.
Let $A_1 = (Q,\Sigma_f,Q \cup \Gamma,\{q_0\},\Delta',Q,o)$
with the following transitions:
\begin{itemize}
	\item $p \xrightarrow{q \mid \eps} p$ for all $p,q \in Q$
	\item $p \xrightarrow{a \mid \gamma} q$ for all $\delta(p,a) = (\gamma,q)$ where $a \in \Sigma_c$
	\item $p \xrightarrow{\tau \mid \eps} \tau(p)$ for all $p \in Q, \tau \in Q^Q$
\end{itemize}
and $o(q) = q$. For each $p \in Q$ let $t_p$ be the rational function defined by
$A_1$ with the only initial state $p$.
One can easily show
that for all $s \in \AllFlat$
we have $t_f(s) = \bot t_{q_0}(s)$
and $t_f(q_1 \cdots q_k s) = \bot t_{q_1}(s)$ for all $q_1 \cdots q_k \in Q^+$.
Hence we can prove that $t_f$ is rational by providing a transducer for $t_f$:
First it verifies whether the input word belongs to the regular language $\AllFlat \subseteq \Sigma_f^*$.
Simultaneously, it verifies whether the input word starts with a state $q \in Q$.
If so, it memorizes $q$ and simulates $A_1$ on $s'$ from $q$,
and otherwise $A_1$ is directly simulated on $s$ from $q_0$.
Since $\rep$ is rational by Lemma~\ref{lem:rep-rat}, $\nu_f$ is also rational.
\end{proof}
If $w = a_1 \cdots a_n \in D$ is a descending word
then $\delta(\bot q_0, w) = \bot p$ for some $p \in Q$.
By definition of $\nu_A$ there exists a state $q \in Q$ with $\nu_A(w) = \bot q$.
Since each suffix of $w$ is also descending
we have $\cev{\nu}_A(w) = \bot q_1 \, \bot q_2 \cdots \bot q_n$ for some $q_1, \dots, q_n \in Q$.
We define $\sigma_0(w) = q_1 \cdots q_n \in Q^*$,
i.e. we remove the redundant $\bot$-symbols from $\cev{\nu}_A(w)$.
If $w$ is non-empty and well-matched we additionally define
$\sigma_1(w) = \tau q_2 \cdots q_n \in Q^Q Q^*$ where $\tau = \varphi(w)$.
We define the sets $S_0 = \sigma_0(D)$
and $S_1 = \sigma_1(W \setminus \{\eps\})$.
Notice that $S_0$ is exactly the set of proper suffixes of words from $S_1$
since descending words are exactly the (proper) suffixes of well-matched words.
We say that $s = s_0 s_1 \cdots s_m \in \AllFlat$ {\em represents} a word $w \in \Sigma^*$
if there exists a monotonic factorization $w = w_0 w_1 \cdots w_m \in \Sigma^*$
such that
$s_0 = \sigma_0(w_0)$,
and for all $1 \le i \le m$ if $w_i$ is well-matched, then $s_i = \sigma_1(w_i)$,
and if $w_i \in \Sigma_c$ then $s_i = w_i$.
Since a word may have different monotonic factorizations,
it may also be represented by many flattenings.
We define the suffix-closed set $\Flat = S_0 (\Sigma_c \cup S_1)^*$,
containing all flattenings which represent some word.

\begin{proposition}
	\label{prop:nu-f}
	If $s \in \AllFlat$ represents $w \in \Sigma^*$ then
	$\nu_f(s) = \nu_A(w)$.
	Therefore, $\nu_f(\Flat) = \Rep$
and $V_L(n) = \log |\cev{\nu}_f(\Flat \cap \Sigma_f^{\le n})|$.
\end{proposition}

\begin{lemma}
	\label{lem:wmf-mon}
	The languages $S_0$ and $S_1$
	are context-free.
\end{lemma}

\begin{proof}
Since $S_0$ is the set of all proper suffixes of words from $S_1$ it suffices to consider $S_1$.
	We will prove that $\{ w \otimes \sigma_1(w) \mid w \in W \}$ is a VPL over the pushdown alphabet
	$(\Sigma_c \times \Sigma_f,\Sigma_r \times \Sigma_f,\Sigma_{\mathit{int}} \times \Sigma_f)$.
	Since the class of context-free languages is closed under projections
	it then follows that $S_1$ is context-free.
	A VPA can test whether the first component $w = a_1 \cdots a_n$ is well-matched
	and whether the second component has the form $\tau q_2 \cdots q_n \in Q^Q Q^*$.
	Since VPLs are closed under Boolean operations, it suffices to test whether
	$\tau \neq \varphi(w)$ or there exists a state $q_i$ with $\nu_A(a_i \cdots a_n) \neq \bot q_i$.
	To guess an incorrect state we use a VPA whose stack alphabet
	contains all stack symbols of $A$ and a special symbol $\#$
	representing the stack bottom.
	We guess and read a prefix of the input word
	and push/pop only the special symbol $\#$ on/from the stack.
	Then at some point we store the second component $q_i$ in the next symbol
	and simulate $A$ on the remaining suffix.
	Finally, we accept if and only if the reached state is $q$ and $\rep(\bot q) \neq \bot q_i$.
	Similarly, we can verify $\tau$ by testing whether there exists a state $p \in Q$ with
	$\varphi(w)(p) \neq \tau(p)$.
\end{proof}

\begin{lemma}
	\label{lem:S-unbounded}
	The language $S_0$ is bounded if and only if $S_1$ is bounded.
	If $S_0$ is not bounded then the $\cev{\nu}_A$-growth of $\Sigma^*$ is exponential
	and therefore $V_L(n) \notin o(n)$.
\end{lemma}

\begin{proof}
Assume that $S_0 \subseteq s_1^* \cdots s_k^*$ is bounded.
Since $S_1 \subseteq \bigcup \{ \tau S_0 \mid \tau \in Q^Q \}$
we have $S_1 \subseteq \tau_1^* \cdots \tau_m^* s_1^* \cdots s_k^*$ for any enumeration $\tau_1, \dots, \tau_m$
of $Q^Q$.
Conversely, if $S_1$ is bounded then each word in $S_0$ is a factor, namely a proper suffix, of a word from $S_1$.
Therefore $S_0$ must be also bounded.

If the context-free language $S_0 = \sigma_0(D) \subseteq Q^*$ is not bounded then its growth must be exponential.
Recall that $\cev{\nu}_A(w)$ and $\sigma_0(w)$ are equal for all $w \in D$ up to the $\bot$-symbol.
Hence $|\cev{\nu}_A(\Sigma^{\le n})| \ge |\cev{\nu}_A(D \cap \Sigma^{\le n})| = |\sigma_0(D \cap \Sigma^{\le n})| = |S_0 \cap Q^{\le n}|$, which proves the growth bound.
\end{proof}

\subparagraph{Bounded overapproximation.}
By Lemma~\ref{lem:S-unbounded} we can restrict ourselves to the case that $S_0$ and $S_1$ are bounded languages,
which will be assumed in the following.
We define $\Psi(a_1 \cdots a_n) = \{ (a_1,n), (a_2,n-1) \dots, (a_n,1) \}$
and $\Psi(L) = \bigcup_{w \in L} \Psi(w)$.

\begin{lemma}
	\label{lem:apx}
	Let $K$ be a bounded context-free language.
	Then there exists a bounded regular superset $R \supseteq K$
	such that $\{ |w| \mid w \in K \} = \{ |w| \mid w \in R \}$ and $\Psi(K) = \Psi(R)$,
	called a {\em bounded overapproximation} of $K$.
\end{lemma}

\begin{proof}
	We use Parikh's theorem \cite{Parikh66}, which implies
	that for every context-free language $K \subseteq \Sigma^*$ the set $\{|w| \mid w \in K\}$ is {\em semilinear},
	i.e. a finite union of arithmetic progressions, and hence
	$\{v \in \Sigma^* \mid \exists w \in K : |v|=|w|\}$ is a regular language.
	Assume that $K \subseteq w_1^* \cdots w_k^*$ for some $w_1, \dots, w_k \in \Sigma^*$.
	We define
	\[
		R = (w_1^* \cdots w_k^*) \cap \{ v \in \Sigma^* \mid \exists w \in K: |v|=|w|  \}
		\cap \{ w \in \Sigma^* \mid \Psi(w) \subseteq \Psi(K) \}.
	\]
	Clearly, $K$ is contained in $R$ and it remains to verify that the third part is regular.
	It suffices to show that for each $a \in \Sigma$ the set $P_a = \{ i \mid (a,i) \in \Psi(K) \}$ is semilinear
	because then an automaton can verify the property $\Psi(w) \subseteq \Psi(K)$.
	Consider the transducer
	\[
		T_a = \{ (a_1 \cdots a_n,\Box^{n-i+1}) \mid a_1 \cdots a_n \in \Sigma^*, \, a_i = a \}.
	\]
	It is easy to see that $T_a$ is rational and $T_aK = \{ \Box^i \mid i \in P_a \}$.
	The claim follows again from Parikh's theorem.
\end{proof}

For each $\tau \in Q^Q$ let $R_\tau$ be a bounded overapproximation of $\tau^{-1} S_1$
and let $R_1 = \bigcup_{\tau \in Q^Q} (\tau R_\tau)$.
Let $R_0 = \bigcup_{\tau \in Q^Q} \Suf(R_\tau)$,
which is the set of all proper suffixes of words in $R_1$.
Both $R_0$ and $R_1$ are also bounded languages.
Finally, set $\RegFlat = R_0 (\Sigma_c \cup R_1)^*$, which is the same as $\Suf((\Sigma_c \cup R_1)^*)$
and is suffix-closed.
According to the definition of bounded overapproximations
we can approximate a word $v = \tau q_2 \cdots q_k \in R_1$ in two possible ways:
Firstly, define $\apx_\ell(v)$ to be any word of the form $\apx_\ell(v) = \tau p_2 \cdots p_k \in S_1$
with $|v| = |\apx_\ell(v)|$.
Secondly, for any position $2 \le i \le k$ define $\apx_i(v)$ to be any word
$\apx_i(v) = \tau s' q_i p_{i+1} \cdots p_k \in S_1$ where $s', p_{i+1} \cdots p_k \in Q^*$.
If $r = r_0 r_1 \cdots r_m \in \RegFlat$ then we can replace any internal factor $r_i \in R_1$
by $\apx_\ell(r_i)$ or any $\apx_j(r_i)$ without changing the value of $\nu_f(r)$.

\begin{proposition}
	\label{prop:flat-reg-eq}
	$\nu_f(\Flat) = \nu_f(\RegFlat) = \Rep$.
\end{proposition}

\begin{proposition}
	\label{prop:transf-growth}
	If $\RegFlat$ contains a linear fooling set for $\nu_f$
	then also $\Flat$ contains a linear fooling set for $\nu_f$.
\end{proposition}

\begin{proof}[Proof of Theorem~\ref{thm:vpl-trichotomy}]

If $L = \emptyset$ or $L = \Sigma^*$ then $V_L(n) \in O(1)$.
Now assume $\emptyset \subsetneq L \subsetneq \Sigma^*$, in which case we have $V_L(n) = \Omega(\log n)$.
Furthermore we know that
$V_L(n) = \log |\cev{\nu}_f(\Flat \cap \Sigma_f^{\le n})|$
by Proposition~\ref{prop:nu-f}.
If the constructed language $S_0$ is not bounded then $V_L(n) \notin o(n)$ by Lemma~\ref{lem:S-unbounded}.
Now assume that $S_0$ is bounded, in which case we can construct the regular language $\RegFlat$.
By Theorem~\ref{thm:rational-dichotomy} the $\cev{\nu}_f$-growth of $\RegFlat$ is either polynomial or exponential
(formally, we have to restrict the domain of $\nu_f$ to the regular language $\RegFlat$).
If the $\cev{\nu}_f$-growth of $\RegFlat$ is polynomial then the same holds for its subset $\Flat$,
and hence $V_L(n) \in O(\log n)$.
If the $\cev{\nu}_f$-growth of $\RegFlat$ is exponential
then by Theorem~\ref{thm:rational-dichotomy} either
the image $\nu_f(\RegFlat)$ is not bounded or $\RegFlat$ contains a linear fooling set for $\nu_f$.
By Proposition~\ref{prop:flat-reg-eq} we have $\nu_f(\RegFlat) = \nu_f(\Flat) = \Rep$.
Hence, if $\Rep$ has exponential growth
then Proposition~\ref{prop:exp-image} implies that $\Flat$ has exponential $\cev{\nu}_f$-growth
and hence $V_L(n) \notin o(n)$.
If $\RegFlat$ contains a linear fooling set for $\nu_f$
then also $\Flat$ contains one by Proposition~\ref{prop:transf-growth}.
By Proposition~\ref{prop:fooling} the $\cev{\nu}_f$-growth of $\Flat$ is exponential
and hence $V_L(n) \notin o(n)$.
\end{proof}

\section{Dichotomy for rational functions}

\label{sec:rat-dichotomy}

In this section we will prove Theorem~\ref{thm:rational-dichotomy}.
Let $t \colon \Sigma^* \to \Omega^*$ be a rational function with
suffix-closed domain $X = \dom(t)$.
By Proposition~\ref{prop:exp-image} the interesting case is where the image $t(X)$ is polynomial growing,
i.e. a bounded language.
There are two further necessary properties in order to achieve polynomial $\cev{t}$-growth. 
Since we apply the rational function to all suffixes,
it is natural to consider right transducers, reading the input from right to left.
The first property states that $t$ has to resemble so called right-subsequential functions,
which are defined by deterministic finite right transducers.
Here we will make use of a representation of rational functions due to Reutenauer and Schützenberger,
which decomposes the rational function $t$ into a right congruence $\R_t$ and a right-subsequential transducer $B$
\cite{ReutenauerS91}.
Secondly, we demand that $B$ is well-behaved,
which means that, roughly speaking, the output produced during a run inside a strongly connected component
only depends on its entry state and the length of the run.
We will prove that in fact these properties are sufficient for the polynomial $\cev{t}$-growth
and in all other cases $X$ contains a linear fooling set.

\subparagraph{The case of finite-index right congruences.}

Let $\sim$ be a finite index right congruence on $\Sigma^*$ and $\approx$ its suffix expansion.
We will characterize those finite index right congruences $\sim$ where $\Sigma^{\le n}/{\approx}$ is polynomially bounded,
which can be viewed as a special case of Theorem~\ref{thm:rational-dichotomy}
since $\nu_\sim \colon \Sigma^* \to \Sigma^*/{\sim}$ is rational.
First assume that $\sim$ is the Myhill-Nerode right congruence $\sim_L$ of a regular language $L$.
Since $\log |\Sigma^{\le n}/{\approx}|$ is exactly the space complexity $V_L(n)$ by Theorem~\ref{thm:psiL},
this case was characterized in \cite{GHKLM18} using so called critical tuples in the minimal DFA for $L$.
We slightly adapt this definition for right congruences.
A {\em critical tuple} in a right congruence $\sim$
is a tuple of words $(u_2, v_2, u, v) \in (\Sigma^*)^4$ such that
$|u_2| = |v_2| \ge 1$,
there exist $u_1, v_1 \in \Sigma^*$ with $u = u_1u_2$, $v = v_1v_2$, and
$u_2 w \not \sim v_2 w$ for all $w \in \{u,v\}^*$.
\begin{proposition}
	\label{prop:ct}
	If $\sim$ has a critical tuple
	then $|\Sigma^{\le n}/{\approx}|$ grows exponentially
	and there exists a critical tuple $(u_2, v_2, u, v)$ in $\sim$ such that $u_2 u \sim u_2 w u$
	and $v_2 u \sim v_2 w u$ for all $w \in \{u,v\}^*$.
\end{proposition}

\begin{proof}
	If $(u_2, v_2, u, v)$ is critical tuple in a right congruence $\sim$
	then we claim that $|\Sigma^{\le n}/{\approx}|$ grows exponentially.
	Let $n \in \N$ and let $w \neq w' \in \{u,v\}^n$.
	There exists a word $z \in \{u,v\}^*$ such that $w$ and $w'$
	have the suffixes $u_2 z$ and $v_2 z$ of equal length.
	By the definition of critical tuples we have $u_2 z \not \sim v_2 z$,
	which implies $w \not \approx w'$.
	Therefore $|\Sigma^{\le cn}/{\approx}| \ge 2^n$ where $c = \max \{|u|,|v|\}$.
	
	The second part is based on the proof of \cite[Lemma~7.4]{GHKLM18a}.
	Let $\equiv$ be the syntactic congruence on $\Sigma^*$
	defined by $x \equiv y$ if and only if $\ell x \sim \ell y$ for all $\ell \in \Sigma^*$.
	Since $\sim$ is a right congruence $\equiv$ is a congruence on $\Sigma^*$ of finite index
	satisfying $\equiv \, \subseteq \, \sim$.
	Define the monoid $M = \Sigma^*/{\equiv}$.
	It is known that there exists a number $\omega \in \N$ such that $m^\omega$ is idempotent
	for all $m \in M$, i.e. $m^\omega \cdot m^\omega = m^\omega$.
	Now let $(u_2, v_2, u, v)$ be a critical tuple and define $u' = (v^\omega u^\omega)^\omega$
	and $v' = (v^\omega u^\omega)^\omega v^\omega$.
	Since $u_2$ is a suffix of $u'$, $v_2$ is a suffix of $u'$ and $u',v' \in \{u,v\}^*$
	the tuple $(u_2,v_2,u',v')$ is again a critical tuple in $\sim$.
	Furthermore we have
	$u'u' = (v^\omega u^\omega)^\omega (v^\omega u^\omega)^\omega \equiv (v^\omega u^\omega)^\omega = u'$
	and $v'u' = (v^\omega u^\omega)^\omega v^\omega (v^\omega u^\omega)^\omega \equiv (v^\omega u^\omega)^\omega = u'$,
	and therefore $u' \equiv wu'$ for all $w \in \{u',v'\}^*$.
	Since $\equiv$ is a congruence this implies
	$u_2 u' \equiv u_2 w u'$ and $v_2 u' \equiv v_2 w u'$ for all $w \in \{u',v'\}^*$,
	and thus also $u_2 u' \sim u_2 w u'$ and $v_2 u' \sim v_2 w u'$,
	which concludes the proof.
\end{proof}

\begin{theorem}
	\label{thm:log}
	Let $L \subseteq \Sigma^*$ be regular.
	Then $V_L(n) \in O(\log n)$ if and only if $|\Sigma^{\le n}/{\approx_L}|$ is polynomially bounded if and only if
	$\sim_L$ has no critical tuple.
\end{theorem}

\begin{proof}
	The first equivalence follows from Theorem~\ref{thm:psiL}.
	By Proposition~\ref{prop:ct} the existence of a critical tuple in $\sim$ implies
	exponential growth of $|\Sigma^{\le n}/{\approx}|$.
	
	Now assume that $V_L(n) \notin O(\log n)$.
	By \cite[Lemma 7.2]{GHKLM18} there exist words $u_2,v_2,u,v \in \Sigma^*$
	such that $u_2$ is a suffix of $u$, $v_2$ is a suffix of $v$, $|u_2| = |v_2|$
	and $u_2 w \not \sim_L v_2 w'$ for all $w,w' \in \{u,v\}^*$
	(one needs the fact that $x \sim_L y$ if and only if $x$ and $y$ reach the same state
	in the minimal DFA for $L$).
	Since in particular $u_2 w \not \sim_L v_2 w$ for all $w \in \{u,v\}^*$
	the tuple $(u_2, v_2, u, v)$ constitutes a critical tuple.
\end{proof}

We generalize this theorem to arbitrary finite index right congruences (Theorem~\ref{thm:log-cong}).
Given equivalence relations $\sim$ and $\sim'$ on a set $X$,
we say that $\sim'$ is {\em coarser} than $\sim$ if $\sim \, \subseteq \, \sim'$,
i.e. each $\sim'$-class is a union of $\sim$-classes.
The {\em intersection} $\sim \cap \sim'$ is again an equivalence relation on $X$.

	\begin{lemma}
	\label{lem:ct-coarse-intersect}
		Let $\sim$ and $\sim'$ be right congruences.
		\begin{enumerate}[(a)]
		\item If $\sim'$ is coarser than $\sim$ and $\sim$ has no critical tuple,
		then $\sim'$ also has no critical tuple.
		\item If $\sim$ and $\sim'$ have no critical tuple
	then $\sim \cap \sim'$ is also a right congruence which has no critical tuple
	\end{enumerate}
	\end{lemma}
	
	\begin{proof}
	Closure under coarsening is clear because the property ``$\sim$ has no critical tuple'' is {\em positive} in $\sim$:
		$\forall u = u_1u_2 \, \forall v = v_1v_2 (|u_2| = |v_2| \to \exists w \in \{u,v\}^* : u_2 w \sim v_2 w)$.
		
	Consider two right congruences $\sim$, $\sim'$ which have no critical tuples.
	One can verify that their intersection $\sim \cap \sim'$ is again a right congruence.
	Let $u = u_1u_2$ and $v = v_1v_2$ with $|u_2| = |v_2|$.
	Because $\sim$ has no critical tuple there exist a word $w \in \{u,v\}^*$ with $u_2 w \mathrel{\sim} v_2 w$.
	Now consider the condition for the words $u_1(u_2w)$ and $v_1(v_2w)$.
	Because $\sim'$ has no critical tuple there exists a word $x \in \{uw, vw\}^*$
	such that $u_2 wx \mathrel{\sim'} v_2 wx$.
	Since $\sim$ is a right congruence we also have $u_2 wx \mathrel{\sim} v_2 wx$
	and thus $u_2 wx \mathrel{(\sim \cap \sim')} v_2 wx$.
	This proves that $\sim \cap \sim'$ has no critical tuple.
\end{proof}

\begin{theorem}
	\label{thm:log-cong}
	$|\Sigma^{\le n}/{\approx}|$ is polynomially bounded if and only if $\sim$ has no critical tuple.
\end{theorem}

\begin{proof}
	Let $u_1, \dots, u_m$ be representatives from each $\sim$-class.
	Observe that $\sim \, = \bigcap_{i=1}^m \sim_{[u_i]_\sim}$
	because $\sim$ saturates each class $[u_i]_\sim$ and
	$\bigcap_{i=1}^m \sim_{[u_i]_\sim}$ also saturates each class $[v]_\sim$.
	Let us write $\sim_i$ instead of $\sim_{[u_i]_\sim}$ and let $\approx_i$
	be its suffix expansion $\approx_{[u_i]_\sim}$.
	Then we have $\sim \, = \bigcap_{i=1}^m \sim_i$ and $\approx \, = \bigcap_{i=1}^m \approx_i$.
	This implies that
	\begin{equation}
		\label{eq:poly-related}
		\max_{1 \le i \le m} |\Sigma^{\le n}/{\approx_i}| \le |\Sigma^{\le n}/{\approx}| \le \prod_{i=1}^m |\Sigma^{\le n}/{\approx_i}|.
	\end{equation}

	\medskip
	\noindent
	($\Rightarrow$): If $|\Sigma^{\le n}/{\approx}|$ is polynomially bounded then the same holds
	for $|\Sigma^{\le n}/{\approx_i}|$ for all $1 \le i \le k$ by \eqref{eq:poly-related}.
	By Theorem~\ref{thm:log} $\sim_{[u_i]_{\sim}}$ has no critical tuple for all $1 \le i \le k$
	and therefore Lemma~\ref{lem:ct-coarse-intersect}(b) implies that $\sim \, = \bigcap_{i=1}^m \sim_{[u_i]_\sim}$
	has no critical tuple.

	\medskip
	\noindent
	($\Leftarrow$): If $\sim$ has no critical tuple then
	each congruence $\sim_i$ has no critical tuple by Lemma~\ref{lem:ct-coarse-intersect}(a)
	because $\sim_i$ is coarser than $\sim$.
	Theorem~\ref{thm:log} implies that $|\Sigma^{\le n}/{\approx_i}|$ is polynomially bounded for all $1 \le i \le k$.
	By \eqref{eq:poly-related} also $|\Sigma^{\le n}/{\approx}|$ is polynomially bounded.
\end{proof}

\subparagraph{Regular look-ahead.}

A result due to Reutenauer and Schützenberger states
that every rational function $f$ can be factorized as $f = r \circ \ell$
where $\ell$ and $r$ are {\em left- and right-subsequential}, respectively \cite{ReutenauerS91}.
A rational function is left- or right-subsequential if the input is read in a deterministic fashion
from left to right and right to left, respectively.
In the literature the order of the directions is usually reversed, i.e. one decomposes $t$ as $f = r \circ \ell$.
Often this is described by the statement that every rational function is (left-)subsequential with regular look-ahead.
Furthermore, this decomposition is canonical in a certain sense.

We follow the notation from the survey paper \cite{FiliotR16}.
A {\em right-subsequential transducer} $B = (Q,\Sigma,\Omega,F,\Delta,\{q_{\mathit{in}}\},o)$ is a real-time
right transducer which is {\em deterministic}, i.e. $q_{\mathit{in}}$ is the only initial state
and for every $p \in Q$ and $a \in \Sigma$ there exists at most
one transition $(p,a,y,q) \in \Delta$.
Clearly, right-subsequential transducers define rational functions, the so called {\em right-subsequential functions},
but not every rational function is right-subsequential.
Let $\R$ be a right congruence on $\Sigma^*$ with finite index.
The {\em look-ahead extension} is the injective function $e_\R \colon \Sigma^* \to (\Sigma \times \Sigma^*/\R)^*$
defined by
\[
	e_\R(a_1 \cdots a_n) = (a_1,[\eps]_\R) (a_2,[a_1]_\R) (a_3,[a_1a_2]_\R) \cdots 
	(a_n,[a_1 \cdots a_{n-1}]_\R).
\]
Let $f \colon \Sigma^* \to \Omega^*$ be a partial function.
The partial function $f[\R] \colon (\Sigma \times \Sigma^*/\R)^* \to \Omega^*$ with $\dom(f[\R]) = e_\R(\dom(f))$
is defined by $f[\R](e_\R(x)) = f(x)$.
Furthermore we define a right congruence $\R_f$ on $\Sigma^*$.
For this we need the distance function $\|x,y\| = |x| + |y| - 2|x \wedge y|$
where $x \wedge y$ is the longest common suffix of $x$ and $y$.
Equivalently, $\|x,y\|$ is the length of the reduced word of $xy^{-1}$ in the free group generated by $\Sigma$.
Notice that $\|\cdot,\cdot\|$ satisfies the triangle inequality.
We define $u \mathrel{\R_f} v$ if and only if
(i) $u \sim_{\dom(f)} v$ and
(ii) $\{ \|f(uw),f(vw)\| \mid uw,vw \in \dom(f) \}$ is finite.
One can verify that $\R_f$ is a right congruence on $\Sigma^*$.
As an example, recall the rational transduction $f$ from Example~\ref{ex:rat-trans}.
The induced right congruence $\R_f$ has two classes, which are $a^*$ and $a^*b\{a,b\}^*$.

\begin{theorem}[\cite{ReutenauerS91}]
	A partial function $f \colon \Sigma^* \to \Omega^*$ is rational if and only if
	$\R_f$ has finite index and $f[\R_f]$ is right-subsequential.
\end{theorem}

For the rest of the section
let $B = (Q,\Sigma \times \Sigma^*/{\R_t},\Omega,F,\Delta,\{q_{\mathit{in}}\},o)$
be a trim right-subsequential transducer for $t[\R_t]$.
One obtains an unambiguous real-time right transducer $A$
for $t$ by projection to the first component,
i.e. $A = (Q,\Sigma,\Omega,F,\Lambda,\{q_{\mathit{in}}\},o)$
where $\Lambda = \{ (q,a,y,p) \mid (q,(a,b),y,p) \in \Delta \}$.
Notice that every run $q \xleftarrow{x \mid y} p$ in $A$
induces a corresponding run $q \xleftarrow{(x,z) \mid y} p$ in $B$ for some $z \in (\Sigma^*/{\R_t})^*$
and that this correspondence is a bijection between the sets of all runs
in $A$ and $B$.
We need two auxiliary lemmas which concern the right congruence $\R_t$.

\begin{lemma}[Short distances]
	\label{lem:linear-distance}
	Let $u,v,w \in \Sigma^*$ with $uw,vw \in X$.
	If $u \mathrel{\R}_t v$
	then $\|t(uw),t(vw)\| \le O(|u|+|v|)$.
\end{lemma}

Two partial functions $t_1, t_2 \colon \Sigma^* \to \Omega^*$ are {\em adjacent} if
$\sup \{\|t_1(w),t_2(w)\| \mid w \in \dom(t_1) \cap \dom(t_2) \} < \infty$
where $\sup \emptyset = -\infty$.
We remark that two functions are adjacent in our definition if and only if
their reversals are adjacent according to the original definition \cite{ReutenauerS91}.
Notice that $u \mathrel{\R_t} v$ if and only if $u \sim_X v$ and
the functions $w \mapsto t(uw)$ and $w \mapsto t(vw)$ are adjacent.

\begin{lemma}[Short witnesses]
	\label{lem:short-witness}
	Let $t_1, t_2 \colon \Sigma^* \to \Omega^*$ be rational functions
	which are not adjacent.
	Then there are words $x,y,z \in \Sigma^*$ such that $xy^*z \subseteq \dom(t_1) \cap \dom(t_2)$
	and $\|t_1(x y^k z),t_2(x y^k z)\| = \Omega(k)$.
	In particular, for each $k \in \N$ there exists a word $x \in \dom(t_1) \cap \dom(t_2)$ of length $|x| \le O(k)$
	such that $\|t_1(x),t_2(x)\| \ge k$.
\end{lemma}

\begin{proposition}
	\label{prop:crit-tup}
	If $\R_t$ has a critical tuple then $X$ contains a linear fooling set.
\end{proposition}

\begin{proof}
	Let $(u_2,v_2,u,v)$ be a critical tuple in $\R_t$ with $u = u_1u_2$ and $v = v_1v_2$.
	By Proposition~\ref{prop:ct} we can assume that $u_2 u \mathrel{\R_t} u_2 w u$ and $v_2 u \mathrel{\R_t} v_2 w u$
	for all $w \in \{u,v\}^*$.
	By assumption we know that $(u_2 u, v_2 u) \notin \R_t$.
	Furthermore, we claim that $u_2 u \sim_X v_2 u$:
	Let $z \in \Sigma^*$ and assume that $u_2 u z \in X$.
	Then $u_2 v_1 v_2 u z \in X$ because $u_2 u \sim_X u_2 v_1 v_2 u$,
	and thus $v_2 u z \in X$ because $X$ is suffix-closed.
	The other direction follows by a symmetric argument.
	
	Let $n \in \N$ and define 
	\[
		N = \max_{x \in \{u_2,v_2\}} \max_{w \in \{u,v\}^{\le n}} \sup \{ \|t(x u z),t(x w u z)\| \mid x u z,x w u z \in X \} < \infty.
	\]
	By Lemma~\ref{lem:linear-distance} we have $N \le O(n)$.
	Since $(u_2 u, v_2 u) \notin \R_t$ and $u_2 u \sim_X v_2 u$,
	the functions $z \mapsto t(u_2 u z)$ and $z \mapsto t(v_2 u z)$ are not adjacent.
	By Lemma~\ref{lem:short-witness} there exists a word $z_n \in (u_2u)^{-1}X$
	with $\|t(u_2 u z_n), t(v_2 u z_n) \| \ge 2N+1$
	and $|z_n| \le O(N) \le O(n)$.
	We claim that $t(u_2 w u z_n) \neq t(v_2 w u z_n)$ for all $w \in \{u,v\}^{\le n}$:
	By the triangle inequality we have
	\begin{align*}
		2N+1 &\le \| t(u_2 u z_n), t(v_2 u z_n) \| \\
		&\le \| t(u_2 u z_n), t(u_2 w u z_n) \| + \| t(u_2 w u z_n), t(v_2 w u z_n) \| + \| t(v_2 w u z_n) , t(v_2 u z_n) \| \\
		&\le 2N + \| t(u_2 w u z_n), t(v_2 w u z_n) \|
	\end{align*}
	which implies $\| t(u_2 w u z_n), t(v_2 w u z_n) \| \ge 1$ and in particular $t(u_2 w u z_n) \neq t(v_2 w u z_n)$.
	We have proved that for each $n \in \N$ there exists a word $z_n$ of length $O(n)$
	such that $t(u_2 w u z_n) \neq t(v_2 w u z_n)$ for all $w \in \{u,v\}^{\le n}$.
	If $Z$ is the set of all constructed $z_n$ for $n \in \N$
	then $\{u_2,v_2\} \{u,v\}^* u Z \subseteq X$ and
	$(u_2,v_2,u,v, uZ)$ is a linear fooling scheme.
\end{proof}

\subparagraph{Well-behaved transducers.}

Let $(Q,\preceq)$ be the quasi-order defined by $q \preceq p$ iff there exists a run from $p$ to $q$ in $A$
or equivalently in $B$.
Its equivalence classes are the strongly connected components (SCCs) of $A$ and $B$.
A word $w \in \Sigma^*$ is {\em guarded} by a state $p \in Q$
if there exists a run $q' \xleftarrow{w} p$ in $A$ such that $p \preceq q'$,
i.e. $p$ and $q'$ belong to the same SCC.
Notice that the set of all words which are guarded by a fixed state $p$ is suffix-closed.
A run $q \xleftarrow{w} p$ in $A$ is {\em guarded} if $w$ is guarded by $p$.
We say that $A$ is {\em well-behaved}
if for all $p \in Q$ and all guarded accepting runs $\pi, \pi'$ from $p$
with $|\pi| = |\pi'|$ we have $\out_F(\pi) = \out_F(\pi')$.
\begin{proposition}
	\label{prop:wb}
	If $A$ is not well-behaved then $X$ contains a linear fooling set.
\end{proposition}

\begin{proof}
	Assume there exist states $p,q,r,q',r' \in Q$,
	and accepting runs $q \xleftarrow{u_2} p$ and $r \xleftarrow{v_2} p$ with $|u_2| = |v_2|$
	and $\out_F(q \xleftarrow{u_2} p) \neq \out_F(r \xleftarrow{v_2} p)$.
	Furthermore let $p \xleftarrow{u_1} q' \xleftarrow{u_2} p$,
	$p \xleftarrow{v_1} r' \xleftarrow{v_2} p$ and
	$p \xleftarrow{s} q_{\mathit{in}}$ be runs.
	Let $u = u_1u_2$ and $v = v_1v_2$
	and consider any word $w \in \{u,v\}^*$.
	Since $t(u_2ws) = \out_F(q \xleftarrow{u_2} p) \, \out(p \xleftarrow{ws} q_{\mathit{in}})$
	and $t(v_2ws) = \out_F(r \xleftarrow{v_2} p) \, \out(p \xleftarrow{ws} q_{\mathit{in}})$,
	we have $t(u_2ws) \neq t(v_2ws)$.
	This shows that $(u_2,v_2,u,v,\{s\})$ is a linear fooling scheme.
\end{proof}
If $\pi$ is a non-empty run $p \xleftarrow{a_1 \cdots a_n} q$ in $A$
and $p \xleftarrow{(a_1,\rho_1) \cdots (a_n,\rho_n)} q$ is the corresponding run in $B$
then we call $\rho_1$ the {\em key} of $\pi$.
The following lemma justifies the name, stating that $\pi$
is determined by the state $q$, the word $a_1 \cdots a_n$ and the key $\rho_1$.

\begin{lemma}
	\label{lem:key}
	If $p \xleftarrow{w} q$ and $p' \xleftarrow{w} q$ are non-empty runs in $A$
	with the same key then the runs must be identical.
\end{lemma}

\begin{proof}
	Assume that $w = a_1 \cdots a_n$ and let $p \xleftarrow{(a_1,\rho_1) \cdots (a_n,\rho_n)} q$
	and $p' \xleftarrow{(a_1,\rho'_1) \cdots (a_n,\rho'_n)} q$ be the corresponding runs in $B$
	with $\rho_1 = \rho_1'$.
	We proceed by induction on $n$.
	If $n = 1$ then this statement is trivial because $B$ is deterministic.
	Now assume $n \ge 2$ and let
	$p \xleftarrow{a_1} r \xleftarrow{a_2 \cdots a_n} q$ and
	$p' \xleftarrow{a_1} r' \xleftarrow{a_2 \cdots a_n} q$.
	Since $B$ is trim there exist an accepting run on $e_{\R_t}(u)$ from $p$
	and an accepting run on $e_{\R_t}(u')$ from $p'$ for some words $u,u' \in \Sigma^*$.
	By definition of $t[\R_t]$ we have $[u]_{\R_t} = \rho_1 = \rho'_1 = [u']_{\R_t}$
	and therefore $\rho_2 = [ua_1]_{\R_t} = [u'a_1]_{\R_t} = \rho_2'$.
	By induction hypothesis we know that the runs
	$r \xleftarrow{a_2 \cdots a_n} q$ and $r' \xleftarrow{a_2 \cdots a_n} q$ are identical.
	Since $p \xleftarrow{(a_1,\rho_1)} r$ and $p' \xleftarrow{(a_1,\rho_1')} r'$
	and $B$ is deterministic we must have $p = p'$.
\end{proof}
Let $\pi$ be any run on a word $y \in \Sigma^*$.
If $\pi$ is not guarded, we can factorize $\pi = \pi' \pi''$
such that $\pi''$ is the shortest suffix of $\pi$ which is unguarded,
and then iterate this process on $\pi'$.
This yields unique factorizations $\pi = \pi_0 \pi_1 \cdots \pi_m$ and $y = y_0 y_1 \cdots y_m$
where $\pi_i$ is a run on $y_i$ from a state $q_i$ to a state $q_{i-1}$
such that $y_i$ is the shortest suffix of $y_0 \cdots y_i$ which is not guarded by $q_i$
for all $1 \le i \le m$ and $\pi_0$ is guarded.
The factorization $\pi = \pi_0 \pi_1 \cdots \pi_m$ is the {\em guarded factorization} of $\pi$.

\begin{proposition}
	\label{prop:poly}
	Assume that $t(X)$ is bounded, $A$ is well-behaved and $\R_t$ has no critical tuple.
	Then the $\cev{t}$-growth of $X$ is polynomially bounded.
\end{proposition}

\begin{proof}

	We will describe an encoding of $\cev{t}(w)$ for $w \in X$ using $O(\log |w|)$ bits.
	For each word $w \in \Sigma^*$ and each state $q \in Q$ we define a tree $T_{q,w}$ recursively,
	which carries information at the nodes and edges.
	If $w$ is guarded by $q$ then $T_{q,w}$ consists of a single node labelled by the pair
	$(q,|w|)$.
	Otherwise let $w = uv$ such that $v$ is the shortest suffix of $w$ which is not guarded by $q$.
	Then $T_{q,w}$ has a root which is labelled by the tuple $(q,|w|,|v|,\cev{\nu}_{\R_t}(u))$.
	For each run $p \xleftarrow{v} q$ we attach $T_{p,u}$ to the root as a direct subtree.
	The edge is labelled by the pair $(\rho,\out(p \xleftarrow{v} q))$
	where $\rho$ is the key of $p \xleftarrow{v} q$.
	By Lemma~\ref{lem:key} distinct outgoing edges from the root are labelled by distinct keys.
	
	The tree $T_{q,w}$ can be encoded using $O(\log |w|)$ bits:
	Since we have $p \prec q$ for every unguarded run $p \xleftarrow{v} q$
	the tree $T_{q,w}$ has height at most $|Q|$ and size at most $|Q|^{|Q|}$.
	All occurring numbers have at most magnitude $|w|$, and the states and keys
	can be encoded by $O(1)$ bits.
	The output words $\out(p \xleftarrow{v} q)$
	are factors of words from the bounded language $t(X)$
	and have length at most $\iml(A) \cdot |v|$.
	Thus they can be encoded using $O(\log |w|)$ bits.
	The node label $\cev{\nu}_{\R_t}(u)$ can be encoded using $O(\log |w|)$ bits
	by Theorem~\ref{thm:log-cong} since $\R_t$ has no critical tuple.

	Let $w = xy \in \Sigma^*$, $q \in Q$
	and let $\pi$ be an accepting run on $y$ from $q$.
	We show that $T_{q,w}$ and $|y|$ determine $\out_F(\pi)$
	by induction on the length of the guarded factorization $\pi = \pi_0 \pi_1 \cdots \pi_m$.
	Since $T_{q_{\mathit{in}},w}$ determines the length $|w|$,
	the tuple $\cev{t}(w)$ is determined by $T_{q_{\mathit{in}},w}$
	for all $w \in X$.
	If $m = 0$ then $y$ is guarded by $q$.
	Since $A$ is well-behaved $\out_F(\pi)$ is determined by $q$
	(which is part of the label of the root of $T_{q,w}$) and $|y|$ only.
	Now assume $m \ge 1$ and suppose that $\pi_i$ is a run $q_{i-1} \xleftarrow{y_i} q_i$
	for all $1 \le i \le m$ with $q_m = q$.
	Then $y_m$ is the shortest suffix of $w$ which is not guarded by $q$.
	The root of $T_{q,w}$ is labelled by $(q,|y_m|,\cev{\nu}_{\R_t}(x y_0 \cdots y_{m-1}))$.
	Since $|y_m|$ and $|y|$ are known, we can also determine $|y_0 \cdots y_{m-1}|$.
	From $\cev{\nu}_{\R_t}(x y_0 \cdots y_{m-1})$ and $|y_0 \cdots y_{m-1}|$
	we can then determine $[y_0 \cdots y_{m-1}]_{\R_t}$, which is the key of $\pi_m$.
	By Lemma~\ref{lem:key} we can find the unique edge
	which is labelled by $([y_0 \cdots y_{m-1}]_{\R_t},\out(\pi_m))$.
	It leads to the direct subtree $T_{q_{m-1},x y_0 \cdots y_{m-1}}$ of $T_{q,w}$.
	By induction hypothesis $T_{q_{m-1},x y_0 \cdots y_{m-1}}$ and $|y_0 \cdots y_{m-1}|$
	determine $\out_F(\pi_0 \cdots \pi_{m-1})$.
	Finally, we can determine $\out_F(\pi_0 \cdots \pi_m) = \out_F(\pi_0 \cdots \pi_{m-1}) \, \out(\pi_m)$,
	concluding the proof.
\end{proof}

Now we can prove Theorem~\ref{thm:rational-dichotomy}:
If $X$ contains no linear fooling set for $t$
then $A$ must be well-behaved by Proposition~\ref{prop:wb}
and $\R_t$ has no critical tuple by Proposition~\ref{prop:crit-tup}.
If additionally $t(X)$ is bounded then the $\cev{t}$-growth of $X$
is polynomially bounded by Proposition~\ref{prop:poly}.
Otherwise, if either $X$ contains a linear fooling set
or $t(X)$ is not bounded then the $\cev{t}$-growth of $X$
is exponential by Proposition~\ref{prop:fooling} and by Proposition~\ref{prop:exp-image}.

\bibliography{refs}

\appendix

\section{Proof of Proposition~\ref{prop:exp-image}}

	Since $\cev{t}(x)$ determines $t(x)$
	we have $|\cev{t}(X \cap \Sigma^{\le n})| \ge |t(X \cap \Sigma^{\le n})|$ for all $n \in \N$.
	It suffices to show that every non-empty preimage $t^{-1}(\{y\})$
	contains at least one word of length $O(|y|)$ in $X$,
	i.e. there exists a number $c > 0$ such that
	$t(X) \cap \Omega^{\le n} \subseteq t(X \cap \Sigma^{\le cn})$ for sufficiently large $n \in \N$.
	Then, if by assumption $|t(X) \cap \Omega^{\le n}|$ grows exponentially,
	then so does $|t(X \cap \Sigma^{\le n})|$ and also $|\cev{t}(X \cap \Sigma^{\le n})|$.
	
	Let us now prove the claim, for which we need to define context-free grammars
	over arbitrary monoids.
	A context-free grammar over a monoid $M$
	has the form $G = (N,S,\to_G)$
	where $N$ is a finite set of nonterminals (which is disjoint from $M$),
	$S$ is the starting nonterminal,
	and $\to_G \, \subseteq N \times (M * N^*)$ is a finite set of productions
	where $M * N^*$ is the free product of the monoids $M$ and $N^*$.
	A derivation tree for $m \in M$
	is a node-labelled rooted ordered tree with the following properties:
	\begin{itemize}
		\item Inner nodes are labelled by nonterminals $A \in N$.
		\item Leaves are labelled by monoid elements $m \in M$.
		\item If a node $s$ has children $s_1, \dots, s_k$
	where $v$ is labelled by $A$ and $s_1, \dots, s_k$ are labelled by $\alpha_1, \dots, \alpha_k$
	then there exists a production $A \to_G \alpha_1 \cdots \alpha_k$.
		\item If $m_1, \dots, m_\ell$ are the labels of the leaves read from left to right
		then $m = m_1 \cdots m_\ell$.
		\end{itemize}
	The language $\L(A)$ generated by a nonterminal $A \in N$ is the set of all elements $m \in M$
	such that there exists a derivation tree for $m$ whose root is labelled by $A$.
	The language $\L(G)$ generated by $G$ is the language $\L(S)$.
	
	We first construct
	from a context-free grammar $G = (N,S,\to_G)$ for $X \subseteq \Sigma^*$
	a context-free grammar $H = (N',S',\to_H)$ for $t|_X = \{ (x,t(x)) \mid x \in X \}$
	over the product monoid $\Sigma^* \times \Omega^*$.
	We can assume that $\eps \notin X$ and that $G$ is in Chomsky normal form,
	i.e. each rule has the form $A \to a$ where $A \in N$ and $a \in \Sigma$,
	or $A \to BC$ where $A,B,C \in N$.
	Let $R = (Q,\Sigma,\Omega,I,\Delta,F,o)$ be a real-time transducer for $t$.
	We define $N' = \{ S' \} \cup \{ A_{p,q} \mid A \in N, \, p,q \in Q \}$
	and $\to_H$ contains the productions
	\begin{itemize}
		\item $A_{p,q} \to_H (a,y)$ for all productions $A \to_G a$
		and transitions $p \xrightarrow{a \mid y} q$ in $R$,
		\item $A_{p,q} \to_H B_{p,r} C_{r,q}$ for all productions $A \to_G BC$ and $p,q,r \in Q$,
		\item $S' \to_H S_{p,q} \, (\eps,o(q))$ for all $(p,q) \in I \times F$.

	\end{itemize}
	One can verify that for all $A \in N$ and $p,q \in Q$ the language $\L(A_{p,q})$ is the set of all pairs
	$(x,y) \in \L(A) \times \Omega^*$ such that $p \xrightarrow{x \mid y} q$ in $R$,
	and that $\L(H) = t|_X$.
	
	Now let $A \in N$, $p,q \in Q$ and $(x,y) \in \L(A_{p,q})$
	with the property that $|x| = \min \{|x'| \mid (x',y) \in \L(A_{p,q}) \}$.
	Consider a derivation tree $T$ for $(x,y)$ whose root is labelled by $A_{p,q}$.
	If $s$ is a node in $T$ which derives $(u,v)$
	then we define the {\em weight} of $s$ to be $|v|$.
	Clearly, the weight of an inner node is the sum of the weights of its children. 

	\begin{claim}
	If $(s_1, s_2, \dots, s_k)$ is a path in $T$ such that all nodes $s_i$ on the path
	have the same length then $k \le |N'|$.
	\end{claim}
	
	\begin{proof}
	Assume that $k > |N'|$.
	There exist two nodes $s_i \neq s_j$ with $i < j$
	which are labelled by the same nonterminal from $N'$.
	The subtrees rooted in $s_i$ and $s_j$ are derivation trees for pairs $(u,v)$ and $(u',v)$
	for some $u,u' \in \Sigma^*$ and $v \in \Omega^*$ where $u'$ is a proper factor of $u$.
	We can then replace the subtree rooted in $s_i$ by the subtree rooted in $s_j$
	and obtain a derivation tree for a pair $(x',y)$ with $|x'| < |x|$, contradiction.
	\end{proof}
	
	Set $c = |N'|$. By the claim above every subtree whose root has weight 0
	has depth at most $c$ and hence its size is at most $C = 2^c-1$.
	Define $D = c+(c-1)C$.

	\begin{claim}	
		The derivation tree $T$ has $O(|y|)$ nodes.
	\end{claim}
	
	\begin{proof}
		We prove by induction on $|y|$ that,
		if $|y| \ge 1$ then $T$ has at most $(2|y|-1)D$ nodes.
		The root of $T$ has weight $|y|$.
		Let $(s_1, \dots, s_k)$ be the maximal path starting in the root whose nodes have weight $|y|$.
		We know that $k \le c$.
		If $s_i'$ is the sibling of $s_{i-1}$ for $2 \le i \le k$,
		then $s_i'$ has weight $0$ and the subtree rooted in $s_i'$ has at most $C$ nodes.
		
		\begin{enumerate}
			\item Assume that $s_k$ is a leaf.
			Then $T$ consists of at most $D = c + (c-1)C \le (2|y|-1)D$ nodes,
			namely $k \le c$ nodes on the path $(s_1, \dots, s_k)$
			and $c-1$ many subtrees with at most $C$ nodes.
			
			\item Assume that $s_k$ has two children $s_{k+1}$ and $s'_{k+1}$
			and let $w$ and $w'$ be the weights of $s_{k+1}$ and $s'_{k+1}$, respectively.
			We have $|y| = w + w'$ and $1 \le w,w' < |y|$.
			By induction hypothesis the subtrees rooted in $s_{k+1}$ and $s'_{k+1}$
			have at most $(2w-1)D$ and $(2w'-1)D$ nodes, respectively.
			Therefore $T$ has in total at most $D + (2w-1)D + (2w'-1)D \le (2|y|-1)D$ nodes.
		\end{enumerate}
		This concludes the proof of the claim.
	\end{proof}
	Now let $y \in t(X) \cap \Omega^{\le n}$ and $x \in X$ be any word with $t(x) = y$.
	There exists an initial accepting run $p \xrightarrow{x \mid y'} q$
	with $y = y' \, o(q)$.
	As shown above there exists a word $x'$ with $p \xrightarrow{x' \mid y'} q$
	and $x' \le O(|y'|) \le O(n)$, which concludes the proof.

\section{Proof of Lemma~\ref{lem:rep-rat}}

In \cite{BaranyLS06} it was observed that $\sim$ can be recognized by a synchronous 2-tape automaton.
The {\em convolution} of two words $u = a_1 \cdots a_m, v = b_1 \cdots b_n \in \Omega^*$
is the word $u \otimes v = c_1 \cdots c_\ell$ of length $\ell = \max(m,n)$ over the alphabet $(\Omega \cup \{\Box\})^2$
where $c_i = (a_i,b_i)$ if $1 \le i \le \min(m,n)$, $c_i = (a_i,\square)$ if $m < i \le n$
and $c_i = (\square,b_i)$ if $n < i \le m$.
Similarly, one can define an associative operation $\otimes$ on $k$-tuples of words.
A $k$-ary relation $R \subseteq (\Omega^*)^k$ is {\em synchronous rational} if $\otimes R = \{ \otimes(u_1, \dots, u_k) \mid (u_1, \dots, u_k) \in R\}$ is a regular language over $(\Omega \cup \{\Box\})^k$.
The set of synchronous rational relations is known to be closed under first-order operations
and, in particular, under Boolean operations, cf. \cite{KhoussainovN94}.
Clearly, every binary synchronous rational relation is a rational transduction.

\begin{lemma}[\cite{BaranyLS06}]
	The equivalence relation $\sim^\rev$ is synchronous rational.
\end{lemma}

\begin{proof}
	We present a right automaton which recognizes the complement of $\sim^\rev$.
	It reads two configurations
	$\alpha p$ and $\beta q$ synchronously
	which are aligned to the right, from right to left.
	The automaton stores a pair of states of $A$, starting with the pair $(p,q)$.
	It then guesses a word $w$ by its monotonic factorization which witnesses
	that $w$ belongs to exactly one of the languages $\L(\alpha p)$ and $\L(\beta q)$.
	Notice that it suffices to read the maximal descending prefix of $w$ and test whether
	the reached state pair $(p',q')$ belongs to some fixed set of state pairs
	since the remaining ascending suffix cannot access the stack contents of the reached configurations.
	To simulate $A$ on a descending prefix in each step the automaton either guesses
	a return symbol and removes the top most stack symbol of both configurations (or leaves $\bot$ at the top),
	or guesses a state transformation $\tau \in \varphi(W)$ which only modifies the current state pair.
\end{proof}

It is well-known that $\le_{\mathit{llex}}$ is a synchronous rational relation.
By the closure properties of synchronous rational relations
the function $\rep$ is rational.

\section{Proof of Proposition~\ref{prop:nu-f}}

Let $w = w_0 w_1 \cdots w_m \in \Sigma^*$ be a monotonic factorization
and let $s = s_0 s_1 \cdots s_m \in \Flat$ be the associated flattening.
We prove $t_f(s) \sim \delta(\bot q_0, w)$ by induction on $m$.
\begin{itemize}
\item If $m = 0$ and $s_0 = \eps$ then $t_f(s) = \bot q_0 = \delta(\bot q_0,\eps)$.
\item If $m = 0$ and $s_0 = q_1 \cdots q_k \in Q^+$
then $t_f(s) = \bot q_1$ and $\nu_A(w) = \rep(\delta(\bot q_0,w)) = \bot q_1$.
\item If $m \ge 1$ and $s_m \in \Sigma_c$ then $s_m = w_m$.
By induction hypothesis we know that $t_f(s_0 \cdots s_{m-1}) \sim \delta(\bot q_0,w_0 \cdots w_{m-1})$.
Since $\delta(\bot q_0,w) = \delta(\delta(\bot q_0,w_0 \cdots w_{m-1}),w_m)$
and $t_f(s) = \delta(t_f(s_0 \cdots s_{m-1}),s_m)$ we obtain $\delta(\bot q_0,w) \sim t_f(s)$.
\item If $m \ge 1$ and $s_m = \tau q_2 \cdots q_k \in Q^Q Q^*$ then $w_m$ is well-matched and $\varphi(w_m) = \tau$.
Assume that $t_f(s_0 \cdots s_{m-1}) = \alpha p$
and $\delta(\bot q_0,w_0 \cdots w_{m-1}) = \beta q$.
By induction hypothesis we know that $\alpha p \sim \beta q$.
Since $t_f(s) = \alpha \tau(p) = \delta(\alpha p, w_m)$
and $\delta(\bot q_0,w) = \delta(\beta q,w_m)$ we obtain $t_f(s) \sim \delta(\bot q_0,w)$.
\end{itemize}
Since $\nu_f = \rep \circ t_f$ and $\nu_A(w) = \rep(\delta(\bot q_0,w))$
we have $\nu_f(s) = \nu_A(w)$.
	
\begin{lemma}
	\label{lem:state-transd}
	Let $w = w_0 w_1 \cdots w_m \in \Sigma^*$ be a monotonic factorization
	with empty initial factor $w_0 = \eps$ and let $s = s_0 s_1 \cdots s_m \in \Sigma_f^*$ be the associated flattening.
	If $\delta(\bot p,w) = \bot \alpha q$
	then $p \xrightarrow{s \mid \alpha} q$ in $A_1$
	and hence $t_p(s) = \alpha q$.
\end{lemma}

\begin{proof}
Proof by induction on $m$. If $m = 0$ then $w = s = \eps$, $p = q$ and $\alpha = \eps$.
For the induction step assume
$\delta(\bot p, w_1 \cdots w_{m-1}) = \bot \alpha q$
and $\delta(\bot \alpha q, w_m) = \bot \alpha \alpha_1 q_1$.
By induction hypothesis the run of $A_1$ on $s$ has the form
$p \xrightarrow{s_1 \cdots s_{m-1} \mid \alpha} q \xrightarrow{s_m \mid \alpha_2} q_2$.
We do a case distinction.

If $w_m \in \Sigma_c$ then $\delta(q,w_m) = (\alpha_1,q_1)$.
Since $s_m = w_m$ and by definition of $A_1$ we have $\alpha_1 = \alpha_2$ and $q_1 = q_2$.
Otherwise $w_m \in W \setminus \{\eps\}$ and $\alpha_1 = \eps$.
The word $s_m = \sigma_1(w_m)$ starts with $\tau = \varphi(w_m)$
and we have $\tau(q) = q_1$.
By definition of $A_1$ we indeed have $q_2 = \tau(q)$ and $\alpha_2 = \eps$.
\end{proof}

We define the following total function $t_f \colon \Sigma_f^* \to (Q \cup \Gamma)^*$.
Let $s \in \Sigma_f^*$ be an input word and let $q_1 \cdots q_k \in Q^*$ be the maximal prefix of $s$
from $Q^*$, say $s = q_1 \cdots q_k s'$ for some $s' \in \Sigma_f^*$.
Then we define
\[
	t_f(s) = \begin{cases}
	t_{q_0}(s), & \text{if $k = 0$,} \\
	t_{q_1}(s'), & \text{if $k \ge 1$.} \\
	\end{cases}
\]
It is easy to see that $t_f$ is rational by providing a transducer for $t_f$.
It verifies whether $s$ starts with a state $q \in Q$.
If so, it memorizes $q$ and simulates $A_1$ on $s'$ from $q$,
and otherwise $A_1$ is directly simulated on $s$ from $q_0$.

Now let $w = w_0 w_1 \cdots w_m$ be a monotonic factorization and $s = s_0 s_1 \cdots s_m \in \Sigma_f^*$
be the associated flattening.
We claim that $\delta(\bot q_0, w) \sim \bot t_f(s)$.
If $w_0 = \eps$ then $s_0 = \eps$ and $s$ does not start with a state from $Q$.
In this case we have $\delta(\bot q_0, w) = \bot  t_{q_0}(s) = \bot t_f(s)$ by Lemma~\ref{lem:state-transd}.
If $w_0 \neq \eps$ then $s_0$ starts with some state $q_1 \in Q$.
By definition of $\sigma_0$ we have $\delta(\bot q_0, w_0) \sim \bot q_1$
and thus $\delta(\bot q_0, w) \sim \delta(\bot q_1, w_1 \cdots w_m)$.
By Lemma~\ref{lem:state-transd} we have $\delta(\bot q_1, w_1 \cdots w_m) = \bot t_{q_1}(s_1 \cdots s_m) = \bot t_f(s)$,
which proves the claim.
Finally, we can set $\nu_f(s) = \rep(\bot t_f(s))$ for all $s \in \Sigma_f^*$.

\section{Proof of Proposition~\ref{prop:flat-reg-eq}}

	By Proposition~\ref{prop:nu-f} we know $\nu_f(\Flat) = \Rep$.
	Clearly $\nu_f(\Flat) \subseteq \nu_f(\RegFlat)$ and it remains to show the other inclusion.
	Consider a word $r \in \RegFlat$ which does not have a non-empty prefix from $R_0$, say
	$r = u_1 v_1 u_2 v_2 \cdots v_m u_{m+1}$
	where $u_1, \dots, u_{m+1} \in \Sigma_c^*$
	and $v_1, \dots, v_m \in R_1$.
	Then $r' = u_1 \, \apx_\ell(v_1) \, u_2 \, \apx_\ell(v_2) \cdots \apx_\ell(v_m) \, u_{m+1}$
	belongs to $\Flat$ and $\nu_f(r) = \nu_f(r')$.
	
	Now assume that $r$ has a non-empty prefix $q_1 \cdots q_k \in R_0$.
	We do the replacements above and the following.
	By definition $q_1 \cdots q_k$ is a proper suffix of some word
	$x = \tau p_2 \cdots p_{i-1} q_1 \cdots q_k \in R_1$.
	Let $y = \apx_i(x) \in S_1$ which has a proper suffix of the form $q_1 q_2' \cdots q_k'$ belonging to $S_0$.
	We can replace $q_1 \cdots q_k$ by $q_1 q_2' \cdots q_k'$ in $r$
	and obtain a word $r' \in \Flat$ with  $\nu_f(r) = \nu_f(r')$.

\section{Proof of Proposition~\ref{prop:transf-growth}}

	Assume that $(u_2,v_2,u,v,Z)$ is a linear fooling scheme for $\nu_f$
	with $\{u_2,v_2\} \{u,v\}^*Z \subseteq \RegFlat$.
	We first ensure that $\{u,v\} \cup Z \subseteq (\Sigma_c \cup R_1)^*$.
	Assume that $u, v \in Q^*$ and hence $\{u_2,v_2\} \{u,v\}^* \subseteq Q^*$ is
	contained in the set of prefixes of words in $R_0$.
	Since $R_0$ is bounded by assumption also $\{u_2,v_2\} \{u,v\}^*$ must be bounded, which contradicts the fact
	that $\{u_2,v_2\} \{u,v\}^*$ has exponential growth.		
	
	Without loss of generality, assume that $u = u_3 u_4$ such that $u_4$ either starts with
	a call letter $a \in \Sigma_c$ or a transformation $\tau \in Q^Q$.
	We claim that $(u_2 u_3,v_2 u_3, u_4 u u_3, u_4 v u_3, u_4 Z)$ is a linear fooling scheme for $\nu_f$.
	It has the following properties:
	\begin{itemize}
	\item $\{ u_2 u_3,v_2 u_3\} \{u_4 u u_3, u_4 v u_3\}^* u_4 Z \subseteq \RegFlat$,
	\item $u_2u_3$ is a suffix of $u_4 u u_3$,
	\item $v_2 u_3$ is a suffix of $u_4 v u_3$,
	\item $|u_2 u_3|=|v_2 u_3|$.
	\end{itemize}
	Also, we know that for every $n \in \N$ there exists a word $z_n \in Z$ with $|z_n| \le O(n)$ and
	$\nu_f(u_2 w z_n) \neq \nu_f(v_2 w z_n)$ for all $w \in \{ uu,uv \}^{\le n}\{u\}$
	and thus, by factoring out the first $u_3$- and the last $u_4$-factor, we have
	$\nu_f(u_2 u_3 w u_4 z_n) \neq \nu_f(v_2 u_3 w u_4 z_n)$ for all $w \in \{ u_4 u u_3,u_4 v u_3 \}^{\le n}$.
	Hence we have verified the conditions of a linear fooling scheme.
	It has the desired properties that
	$\{u_4 u u_3, u_4 v u_3\} \cup u_4 Z \subseteq (\Sigma_c \cup R_1)^*$
	because $u_4$ starts with a call letter or a transformation.
	
	Now let $(u_2,v_2,u,v,Z)$ be a linear fooling scheme with $\{u,v\} \cup Z \subseteq (\Sigma_c \cup R_1)^*$.
	We replace occurring factors from $R_1$ by factors from $S_1$
	while maintaining the values $\nu_f(u_2wz)$ and $\nu_f(v_2wz)$ for $w \in \{u,v\}^*$ and $z \in Z$.
	
	\begin{enumerate}
	\item First, in each word $z \in Z \subseteq (\Sigma_c \cup R_1)^*$ we replace
	each $R_1$-factor $v$ by $\apx_\ell(v)$ which ensures that $Z \subseteq (\Sigma_c \cup S_1)^*$.

	\item Next consider $u$ and $v$, and
	assume that $u = u_1 u_2$ and $v = v_1v_2$ for some $u_1,v_1 \in \Sigma_f^*$.
	Let us consider $R_1$-factors which cross the factorization $u = u_1 u_2$ or $v = v_1 v_2$, respectively.
	If $u_2$ starts with some state we can factorize $u_1$ and $u_2$ as
	$u_1 = u_3 \tau q_2 \cdots q_{i-1}$ and $u_2 = q_i \cdots q_k u_4$ 
	where $u_3, u_4 \in (\Sigma_c \cup R_1)^*$ and $\tau q_2 \cdots q_k \in R_1$.
	Let $\apx_i(\tau q_2 \cdots q_k) = \tau s' q_i p_{i+1} \cdots p_k \in S_1$.
	We replace $u_1$ by $u_3 \tau s'$ and $u_2$ by $q_i p_{i+1} \cdots p_k u_4$.
	Notice that the length of $u_2$ has not changed (this maintains $|u_2| = |v_2|$)
	and the first state of $u_2$ has not changed either (this maintains the values $\nu_f(u_2 w z)$).
	If $v_2$ starts with some state we do the analogous replacements for $v_1$ and $v_2$.
	
	\item Finally, each $R_1$-factor $v$ in $u_1$, $u_2$, $v_1$ and $v_2$
	is replaced by $\apx_\ell(v)$.
	\end{enumerate}
	
	One can verify that the obtained tuple $(u_2,v_2,u,v,Z)$ is again a linear fooling scheme
	for $\nu_f$ satisfying $\{u_2,v_2\}\{u,v\}^* Z \subseteq \Flat$.

\section{Proof of Lemma~\ref{lem:linear-distance}}

	Suppose that $w = a_1 \cdots a_m$.
	Since $\R_t$ is a right congruence we know that $ua_1 \cdots a_i \mathrel{\R}_t v a_1 \cdots a_i$
	for all $0 \le i \le m$.
	By definition of the look-ahead extension
	the words $e_{\R_t}(uw)$ and $e_{\R_t}(vw)$ have the common suffix
	\[
		s = \binom{a_1}{[u]_{\R_t}} \binom{a_2}{[ua_1]_{\R_t}} \cdots \binom{a_m}{[u a_1 \cdots a_{m-1}]_{\R_t}}.
	\]
	The initial accepting runs of $B$ on $e_{\R_t}(uw)$ and $e_{\R_t}(vw)$ have the form
	\[
		q \xleftarrow{e_{\R_t}(u)} p \xleftarrow{s} q_{\mathit{in}}
		\quad \text{and} \quad
		r \xleftarrow{e_{\R_t}(v)} p \xleftarrow{s} q_{\mathit{in}}
	\]
	and thus $t(uw)$ and $t(vw)$ share the suffix $\out(p \xleftarrow{s} q_0)$.
	This implies
	\[
		\|t(uw),t(vw)\| \le |\out_F(q \xleftarrow{e_{\R_t}(u)} p)| + |\out_F(r \xleftarrow{e_{\R_t}(v)} p)| \le \iml(A) \cdot (|u|+|v|+2),
	\]
	proving the statement.

\section{Proof of Lemma~\ref{lem:short-witness}}

	Assume that $t_1$ and $t_2$ are not adjacent.
	By \cite[Proof of Proposition 1.]{ReutenauerS91}
	there exist words $x,y,z \in \Sigma^*$ and $u_1,u_2,v_1,v_2,w_1,w_2 \in \Omega^*$ such that
	$t_1(x y^k z) = u_1v_1^kw_1$, $t_2(x y^k z) = u_2v_2^kw_2$ for all $k \in \N$, and
	$\sup \{ \|u_1v_1^kw_1,u_2v_2^kw_2\| \mid k \in \N \}= \infty$.
	By the triangle inequality we have 
	\begin{align*}
		\|v_1^kw_1,v_2^kw_2\| &\le \|v_1^kw_1,u_1v_1^kw_1\| + \|u_1v_1^kw_1,u_2v_2^kw_2\| + \|u_2v_2^kw_2,v_2^kw_2\| \\
		&= |u_1| + \|u_1v_1^kw_1,u_2v_2^kw_2\| + |u_2| = \|t_1(x y^k z),t_2(x y^k z)\| + |u_1| + |u_2|.
	\end{align*}
	We prove that $\|v_1^kw_1,v_2^kw_2\| \ge \Omega(k)$
	which implies that $\|t_1(x y^k z),t_2(x y^k z)\| \ge \Omega(k)$.
	If both $v_1 = v_2 = \eps$ then
	\[
		\sup_{k \in \N} \|v_1^kw_1,v_2^kw_2\| = \|w_1,w_2\| < \infty,
	\]
	which contradicts $\sup \{ \|u_1v_1^kw_1,u_2v_2^kw_2\| \mid k \in \N \}= \infty$.
	If $|v_1| \neq |v_2|$ then
	\[
		\|v_1^kw_1,v_2^kw_2\| \ge \big| |v_1^kw_1| - |v_2^kw_2| \big| = \Omega(k).
	\]
	Now assume $|v_1| = |v_2| \ge 1$. Since
	\[
		\|v_1^kw_1,v_2^kw_2\| = |v_1^kw_1| + |v_2^kw_2| - 2|v_1^kw_1 \wedge v_2^kw_2| \ge \Omega(k)
		-2|v_1^kw_1 \wedge v_2^kw_2|
	\]
	it suffices to show that $\sup_k |v_1^kw_1 \wedge v_2^kw_2| < \infty$.
	Towards a contradiction assume that $\sup_k |v_1^kw_1 \wedge v_2^kw_2| = \infty$.
	Then, for every $k \in \N$ there exists $K \in \N$ such that
	$|v_1^Kw_1 \wedge v_2^Kw_2| \ge \max \{ |v_1^kw_1|,|v_2^kw_2| \}$.
	If $|v_1^kw_1| \ge |v_2^kw_2|$ then $v_1^k w_1$ is a suffix of $v_1^Kw_1 \wedge v_2^Kw_2$
	and otherwise $v_2^k w_2$ is a suffix of $v_1^Kw_1 \wedge v_2^Kw_2$.
	This shows that for all $k \in \N$ either $v_1^k w_1$ is a suffix of $v_2^kw_2$, or vice versa,
	and therefore $|v_1^kw_1 \wedge v_2^kw_2| = \min \{|v_1^kw_1|,|v_2^kw_2|\}$.
	Since $|v_1| = |v_2|$ we obtain
	\[
		\|v_1^kw_1,v_2^kw_2\| = |v_1^kw_1| + |v_2^kw_2| - 2\min \{|v_1^kw_1|,|v_2^kw_2|\}
		= |w_1| + |w_2| - 2 \min\{|w_1|,|w_2|\}
	\]
	contradicting $\sup_k \|v_1^kw_1,v_2^kw_2\| = \infty$.

\end{document}